\documentclass[5p]{elsarticle}
\pdfoutput=1
\usepackage{graphicx,color}
\usepackage{amsmath,mathtools}
\usepackage{amssymb}
\usepackage{mathrsfs}
\usepackage[vlined,ruled]{algorithm2e}    % creates problem with autart.cls
\usepackage{subfigure}

\usepackage{dsfont}
\usepackage{booktabs}
\usepackage{array}
\usepackage{hyperref}

\usepackage{enumitem}

\newtheorem{theorem}{Theorem}[section]
\newtheorem{lemma}[theorem]{Lemma}
\newtheorem{definition}{Definition}

\newtheorem{remark}{Remark}

\newproof{proof}{Proof}

\newcommand\ajustspaceandequationnumber{%
   \vspace{-\belowdisplayskip}
   \vspace{-\abovedisplayskip}
   \addtocounter{equation}{-1}}

\newcommand{\real}{\mathbb{R}}

\newcommand{\trans}{\mathsf{T}} %or \top or \intercal
\newcommand{\conjtrans}{\mathsf{H}}

\newcommand{\mc}{\mathcal}

\newcommand{\1}{\mathds{1} }

\DeclareSymbolFont{bbold}{U}{bbold}{m}{n}
\DeclareSymbolFontAlphabet{\mathbbold}{bbold}

\newcommand\oprocendsymbol{\hbox{$\square$}}
\newcommand\oprocend{\relax\ifmmode\else\unskip\hfill\fi\oprocendsymbol}

\def\QEDopen{{\setlength{\fboxsep}{0pt}\setlength{\fboxrule}{0.2pt}\fbox{\rule[0pt]{0pt}{1.3ex}\rule[0pt]{1.3ex}{0pt}}}}
\def\QED{\QEDopen} % default to closed

\begin{document}
\begin{frontmatter}
\title{{\bf On the Real Stability Radius of Sparse Systems}}

\author{Vaibhav Katewa and Fabio Pasqualetti} 
\address{Department of Mechanical Engineering, University of California, Riverside, CA, USA}
\ead{\{vkatewa,fabiopas\}@engr.ucr.edu}
%\thanks{This material is
%    based upon work supported in part by ARO award 71603NSYIP and
%    NSF award ECCS1405330. The authors are
%    with the Department of Mechanical Engineering, University of
%    California at Riverside.
%    %\href{mailto:vkatewa@engr.ucr.edu}{\texttt{vkatewa@engr.ucr.edu}}
%    \href{mailto:vkatewa@engr.ucr.edu}{\texttt{\{vkatewa,}}
%    \href{mailto:fabiopas@engr.ucr.edu}{\texttt{fabiopas\}}}
%    \texttt{@engr.ucr.edu} }

\begin{abstract}
In this paper, we study robust stability of sparse LTI  systems using the stability radius (SR) as a robustness measure. We consider real perturbations with an arbitrary and pre-specified sparsity pattern of the system matrix and measure their size using the Frobenius norm. We formulate the SR problem as an equality-constrained minimization problem. Using the Lagrangian method for optimization, we characterize the optimality conditions of the SR problem, thereby revealing the relation between an optimal perturbation and the eigenvectors of an optimally perturbed system. Further, we use the Sylvester equation based parametrization to develop a penalty based gradient/Newton descent algorithm which converges to the local minima of the optimization problem. Finally, we illustrate how our framework provides structural insights into the robust stability of sparse networks.
%We also extend our optimization framework to obtain approximately-sparse optimal perturbations of the SR problem, for any given arbitrary sparsity pattern. 
\end{abstract}

\end{frontmatter}

\section{Introduction}
Guaranteed stability under parameter uncertainty is one of the central problems in robust design of dynamical systems. Consider the following Linear Time-Invariant (LTI) dynamical system
\begin{align}  \label{eq:dynamics}
\mathcal{D} x(t) = Ax(t),
\end{align}
where $x\in\real^{n}$ is the state, $A\in\real^{n\times n}$, and $\mc{D}$ can either be the continuous-time differential operator (i.e., $\mc{D}x(t) = \dot{x}(t)$) or the discrete-time shift operator (i.e., $\mc{D}x(t) = x(t+1)$). 
Let the complex plane $\mathbb{C}$ be divided into any two disjoint sets as $\mathbb{C} = \mathbb{C}_g \cup \mathbb{C}_b$, where $\mathbb{C}_g$ is open. 
Stability of \eqref{eq:dynamics} requires the eigenvalues of $A$ to lie in the stability region $\mathbb{C}_g$.
%(left-half plane and unit disc for continuous and discrete-time systems, respectively). 
Assuming that $A$ is stable, the robust stability analysis of \eqref{eq:dynamics} involves characterizing the eigenvalues of the affine perturbations of $A$ given by 
\begin{align} \label{eq:pert_str}
A \rightsquigarrow A(\Delta) \triangleq A+B\Delta C, % \qquad \Delta \in \mc{S},
\end{align} 
where $\Delta \in \real^{m\times p}$ is the perturbation matrix and
$B \in \real^{n\times m}$, $C \in \real^{p\times n}$ are structure
matrices. The perturbed matrix $A(\Delta)$ can also be interpreted
as the closed loop matrix of the following linear system
\begin{align*}
\mc{D}x(t) &= Ax(t) + Bu(t),\\
y(t) &= Cx(t),
\end{align*}
with static feedback $u(t) = \Delta y(t)$.  

There have been numerous studies on eigenvalue characterization and stability of the perturbed matrix \eqref{eq:pert_str} (see \cite{DH-AJP:05} for a comprehensive treatment). However, an inherent \emph{crucial assumption} in these studies (and of the robust stability theory) is that all entries of the perturbation $\Delta$ are allowed to be freely perturbed. Clearly, this assumption is not applicable in modern control systems which are increasingly networked and distributed in nature and, as a result, exhibit a specific sparsity structure. In such systems, the matrix $A$ typically has an associated sparsity pattern, and its certain entries are fixed/zero, and it is feasible to perturb only the non-fixed entries of $A$. Therefore, the perturbations $\Delta$ applied to $A$ cannot be chosen freely and must satisfy certain sparsity constraints as well. 

In this paper, we develop a novel framework to study the robust stability of LTI systems with sparsity constraints. Let $S\in \{0,1\}^{m\times p}$ be a binary matrix that specifies the sparsity structure of the perturbation $\Delta$. Specifically,
\begin{align*}
\Delta_{ij} = 
\begin{cases}
0  \quad \text{if} \:\:S_{ij} = 0, \text{and} \\
\star \quad \text{if} \:\:S_{ij} = 1,
\end{cases}
\end{align*}
where $\star$ denotes any real number. Further, let $\mathbf{\Delta}_S$ denote the set of sparse perturbations given by
\begin{align} \label{eq:sparse_pert_set}
\mathbf{\Delta}_S =\{\Delta \in \real^{m\times p}: S^{c} \circ  \Delta  = 0\},
 \end{align} 
 where $S^{c} \triangleq 1_{m\times p} - S$ denotes the complementary sparsity structure matrix. 
We consider the notion of Stability Radius (SR) as the measure of robust stability, which is the minimum-size real perturbation that moves an eigenvalue(s) of $A(\Delta)$ outside the stability region. Formally, the SR is defined as
\begin{align} \label{eq:stab_rad}
r &\triangleq \inf\: \{\lVert \Delta \rVert : \Gamma(A(\Delta)) \cap \mathbb{C}_b \neq \emptyset, \:\Delta \in \mathbf{\Delta}_S\},
\end{align}
where $\Gamma(\cdot)$ denotes the spectrum of a matrix. The SR provides a worst case measure for the robustness of the system in the sense that all perturbations with $||\Delta||< r$ are guaranteed to preserve the stability of the perturbed system. 
Note that while matrices $B$ and $C$ can be chosen to impose certain perturbation structures on $A$ (such as zero column(s) or row(s)), they cannot capture arbitrary sparsity constraints. Thus, we require the explicit sparsity constraint $\Delta \in \mathbf{\Delta}_S$ in \eqref{eq:stab_rad}.

The perturbation size $\lVert\Delta \rVert$ can be measured using the spectral norm or Frobenius norm, and these respective cases are referred to as $2$-norm SR and $F$-norm SR. The case when $\Delta$ is allowed to be complex is referred to as complex SR. The case $B=C=I$ and $\mathbf{\Delta}_S = \real^{m\times p}$, where each entry of $A$ is allowed to be perturbed independently, is referred to as unstructured SR.  
In this paper, we study the sparse, real, F-norm SR problem by formulating it as an equality-constrained optimization problem. The real SR problem is more suitable for engineering applications where the dynamics matrix $A$ and its perturbations are typically real. Further, unlike the spectral norm, the Frobenius norm explicitly measures the entry-wise perturbations of $\Delta$, which is useful in characterizing the \emph{size and structure} of sparse perturbations. 
%It also allows for analytical computation of the gradient of the cost of the optimization problem, which we use  to develop an iterative gradient descent algorithm to obtain a local solution. 

\noindent\textbf{Related work} The stability radius problem without sparsity constraints has a rich history. Although robust stability problems have been studied in various forms in the past, the notion of $2$-norm  stability radius was introduced formally in \cite{DH-AJP:86,DH-AJP:86b}. Various bounds and characterizations of unstructured, complex, and real SRs were given in \cite{DH-AJP:86,CVL:85}. In \cite{DH-AJP:86b}, characterizations of structured, complex SR was presented in terms of the $H_{\infty}$-norm of the associated transfer function and solutions of parametrized algebraic Riccati equations. Bisection algorithms to compute the complex SR were presented in \cite{RB:88,DH-BK-AL:89} and algorithms to compute the $H_{\infty}$-norm were given in \cite{SB-VB-PK:89,SB-VB:90,NAB-MS:90}. Since the optimal perturbation for the complex SR problem is always rank-$1$ \cite{DH-AJP:05}, the $2$-norm and the $F$-norm SRs are equal for the complex case.

The real SR problem is considerably more difficult than its complex counterpart \cite{DH-AJP:86}. Qiu et.al. presented several lower bounds for the unstructured case in \cite{LQ-EJD:91,LQ-EJD:92} and a complete algebraic characterization of the structured case was presented in \cite{LQ-BB-AR-EJD-PMY-JCD:95}. Based on this characterization, a level-set algorithm was developed in \cite{JS-PVD-ALT:96} for the structured case and an implicit determinant method was provided for the unstructured case in \cite{MAF-AS:14}. For a comprehensive treatment of the $2$-norm SR problem, see \cite[Chapter~5]{DH-AJP:05}.

While the $2$-norm SR problem has been studied extensively, there are limited studies on the $F$-norm SR problem. Note that due to the fundamental difference between the spectral and Frobenius norms, the procedure in \cite{LQ-BB-AR-EJD-PMY-JCD:95} to characterize the real, $2$-norm SR cannot be used to characterize the real, $F$-norm SR. In \cite{NAB-AVB-PD:99} and \cite{NAB-AVB-PD:01}, lower bounds on the real, $F$-norm SR were provided for the unstructured and structured cases, respectively. Recently, a number of works have appeared that use iterative algorithms to approximate the $2$-norm/$F$-norm real SR \cite{NG-MM:15,MWR:15,NG:16,NG-CL:13,NG-MG-TM-MLO:17}. Typically, these algorithms use two levels of iterations. The inner iteration approximates the rightmost (outermost, for discrete-time case) points of spectral value sets, and the outer iteration verifies the intersection of these points with the stability boundary. All of these aforementioned studies consider the non-sparse SR problem.

In a very recent paper \cite{SCW-MW-MZ-RAD:18}, which was developed
independently and concurrently with our paper, the authors study the
structured distance of an LTI system from the set of systems that do
not exhibit a general property $\mathcal{P}$, such as controllability,
observability, and stability. They provide necessary conditions for a
locally optimal perturbation and develop an algorithm to obtain such
solution. Since the framework in \cite{SCW-MW-MZ-RAD:18} is developed
for a general class of problems, the provided necessary conditions are
implicit in terms of abstract linear maps. On the other hand, we use a
different approach to obtain stronger and explicit necessary
conditions for the sparse SR problem. In addition, we provide
sufficient conditions for a local minimum, thereby completely 
characterizing the local minima. Further, the paper
\cite{SCW-MW-MZ-RAD:18} makes a crucial assumption on the surjectivity
of a linear map that limits the number of allowed sparsity
constraints. Our framework does not have such a limitation, and is
valid for a larger class of sparsity constraints. Finally, our
gradient/Newton descent algorithm is simpler than the algorithm
proposed in~\cite{SCW-MW-MZ-RAD:18}.

%While all these studies have considerably advanced the theory of robust stability, to the best of our knowledge, there have been no prior studies on the characterization of the SR in \emph{the presence of sparsity constraints}. Sparsity constraints introduce additional challenges and the methods developed for computing the non-sparse SR are inapplicable in this case, thereby requiring a new approach and framework.
%Of particular interest is \cite{NG-MG-TM-MLO:17}, which, to the best of our knowledge, is the only study of the real, structured, $F$-norm SR problem. Using techniques from the authors' previous works, this article presents an expansion-contraction algorithm to approximate the SR. This algorithm involves finding locally rightmost points of spectral value sets by iteratively solving an ordinary differential equation (ODE). However, these ODE iterations are not guaranteed to converge to the locally rightmost points and, as a result, the algorithm only provides an upper bound to the SR.\footnote{Although the authors claim that the iterations typically converge in their simulation studies and the bound is tight.} 

\noindent \textbf{Contributions} The contribution of the paper is two-fold. First, we propose a novel approach to compute the sparse SR by formulating the SR problem as an equality constrained minimization problem. We characterize its local optimality conditions, thereby revealing important geometric properties of the optimal perturbed system.  
%Interestingly, our local optimality conditions are the same as the local optimality conditions for the rightmost point of the spectral value sets in \cite{NG-MG-TM-MLO:17}, indicating an underlying connection between our solution and the geometry of spectral value sets.
Second, using the Sylvester equation based parametrization, we develop a penalty-based gradient algorithm to solve the optimization problem that is guaranteed to converge to a local minima. Numerical studies are included to illustrate various properties of the optimal perturbations and the algorithm, and to highlight the usefulness of the framework for sparse networks.

%Further, we extend our optimization problem by suitably weighing the perturbation entries, which allows us to obtain approximately-sparse solutions to the SR problem for arbitrary sparsity constraints on the perturbation. This further highlights the importance of our optimization based approach and the usage of the Frobenius norm in solving sparse stability problems. To the best of our knowledge, our study is the first to consider the SR problem with sparsity constraints on the perturbations. 

\noindent \textbf{Paper organization} In Section \ref{sec:opt_prob} we present our mathematical notation and some properties that we use in the paper, and formulate the SR problem as an  optimization problem with equality constraints. Section \ref{sec:soln_opt_prob} contains the local optimality conditions of the SR problem. In Section \ref{sec:grad_algo} we develop a gradient based algorithm to compute local solutions. Numerical examples are presented in Section \ref{sec:simulations}. Finally, we conclude the paper in Section \ref{sec:conclusion}.  

%\noindent \textbf{Contributions} We study the real, structured SR problem using the Frobenius norm. The contribution of the paper is threefold. First, we cast the SR problem as an equality-constrained minimization problem and characterize its optimality conditions, thereby revealing some important properties of the optimal perturbed system and its eigenvectors. Second, we develop a gradient based algorithm to solve the SR optimization problem that is guaranteed to converge to a local minima. Third, we present a modification of the original optimization problem that allows us to obtain approximately-sparse solutions of the SR problem.   

\section{Problem Formulation}  \label{sec:opt_prob}
\subsection{Mathematical Notation and Properties}
We use the following notation throughout the paper:
$\Vert\cdot\Vert_F$ and $\Vert\cdot\Vert_2$ denote the Frobenius and spectral norm of a matrix, respectively. $\circ$ and $\otimes $ denote the Hadamard and Kronecker product, respectively. The identity matrix is denoted by $I$. $\Gamma(\cdot)$, $(\cdot)^{\trans}$ and tr($\cdot$) denote the spectrum, transpose and trace of a matrix, respectively. $(\cdot)^{+}$ and $\alpha(\cdot)$ denote the pseudo-inverse and spectral abscissa of a matrix, respectively. A positive-definite matrix $A$ is denoted by $A>0$. $(\cdot)^{*}$ and $(\cdot)^{\conjtrans}$ denote the complex conjugate and the conjugate transpose, respectively. Re($\cdot$) and Im($\cdot$) denote the real and imaginary parts of a complex number, respectively. vec($\cdot$) denotes the vectorization of a matrix. diag($a$) denotes a $n\times n$ diagonal matrix with diagonal elements given by $a\in\real^{n}$. $1_{m\times n}$ denotes a $m\times n$ matrix of all ones. Finally, $j = \sqrt{-1}$ denotes the unit imaginary number.

%\vspace*{-10pt}
%\section{Problem Formulation}  \label{sec:opt_prob}
%\vspace*{-10pt}
%\subsection{Mathematical Notation and Properties}
%We use the following notation throughout the paper:
%\begin{tabular} {|c|l|}
%\hline
%$\Vert\cdot\Vert_\text{F}$ & Frobenius norm\\
%$\circ$ & Hadamard (element-wise) product \\
%$\otimes $ & Kronecker product \\
%%$1_n (0_n)$ & $n$-dim vector of ones (zeros)\\
%%$1_{n\times m}(0_{n\times m})$ & $n\times m$-dim matrix of ones (zeros)\\
%$I$ & Identity matrix\\
%$\Gamma(\cdot)$ & Spectrum of a matrix\\
%$(\cdot)^{+}$  & Moore-Penrose pseudo inverse \\
%$(\cdot)^{\trans}$ & Transpose of a matrix\\
%%$\Vert\cdot\Vert_2$ & spectral norm \\
%%$\lvert . \rvert$ & Cardinality of a set \\
%tr($\cdot$)  & Trace of a matrix \\
%$\alpha(\cdot)$ & Spectral abscissa of a matrix\\
%%$e_i$ & $i$-th canonical vector\\
%$A>0$ & Positive-definite matrix $A$\\
%$(\cdot)^{*}$ & Complex conjugate\\
%$(\cdot)^{\conjtrans}$ &  Conjugate transpose \\
%Re($\cdot$) & Real part of a complex variable \\
%Im($\cdot$) & Imaginary part of a complex variable\\
%%$\mathcal{R}$ & Range of a matrix\\ 
%vec($\cdot$) & Vectorization of a matrix\\
%diag($a$) & $n\times n$ diagonal matrix with diagonal\\ 
% & elements given by $a\in\real^{n}$\\
% $j$ & $\sqrt{-1}$ \\
%%$<.,.>_F$ & Inner (Frobenius) product\\ 
%\hline
%\end{tabular}

%The Kronecker sum of two square matrices $A$ and $B$ with dimensions $n$ amd $m$, respectively, is denoted by
%\begin{align*}
%A\oplus B = (I_m \otimes A) + B\otimes I_n.
%\end{align*}
We use the following mathematical properties for the derivation of our results \cite{KBP-MSP:12}, \cite{JRM-HN:99}:
\begin{enumerate} [label=P.\arabic*, align=left]
%\item tr($A$) = tr($A^\trans$)  and tr($ABC$) = tr($CAB$), \label{prop:trace}
\item $\Vert A \Vert_{ \text{F}}^{2} = \text{tr}(A^\trans A) = \text{vec}^\trans(A)\text{vec}(A)$, \label{prop:frob}
\item $\text{vec}(AB) = (I\otimes A)\text{vec}(B) = (B^{\trans} \otimes I)\text{vec}(A)$, \label{prop:vec1}
\item $\text{vec}(ABC) = (C^{\trans}\otimes A)\text{vec}(B)$, \label{prop:vec2}
\item $(A\otimes B)^{\trans} = A^{\trans} \otimes B^{\trans}$ and $(A\otimes B)^{\conjtrans} = A^{\conjtrans} \otimes B^{\conjtrans}$, \label{prop:kron}
\item $(A\otimes B)(C\otimes D) = (AC\otimes BD)$ and \\ \hspace*{12pt}$(A\otimes B)^{+} = A^{+} \otimes B^{+}$, \label{prop:kron1}  
\item $\text{vec}(A\circ B) = \text{vec}(A)\circ \text{vec}(B),  (A\circ B)^\trans = A^\trans\circ B^\trans$, \label{prop:had2} 
\item $\frac{d}{dX}\text{tr}(AX) \!= \!A^\trans$, $\frac{d}{dX}\text{tr}(X^\trans X) \!=\! 2X$, $\frac{d}{dx}(Ax)\! =\! A$, \label{prop:der} 
\item Let $D_xf$ and $D^2_xf$ be the gradient and Hessian of \\ \noindent $f(x): \mathbb{R}^{n}\rightarrow \mathbb{R}$. Then, $df = (D_xf) ^{\trans} dx$ and, \\ \noindent $d^2f = (dx)^{\trans}(D^{2}_xf) dx$, \label{prop:grad_Hess} 
\item $\text{vec}(A^{\trans}) = T_{m,n}\text{vec}(A)$, where $A\in\real^{m\times n}$ and\\ \noindent $T_{m,n} \in \{0,1\}^{mn \times mn}$ is a binary permutation matrix. \label{prop:per_mat} 
%\item  $d(X^{+}) =(I-X^{+}X)(dX)^{\trans}(X^{+})^{\trans} X^{+} \\
  %\hspace*{35pt}+ X^{+} (X^{+})^{\trans} (dX)^{\trans} (I-X^{+}X) -X^{+} dX X^{+}$.  \label{prop:diff_pinv}
\end{enumerate}

\subsection{Sparse Stability Radius as an Optimization Problem}
In this subsection we formulate the real, sparse, $F$-norm SR problem as an equality-constrained optimization problem. Due to space limitations, we present the analysis only for continuous-time systems in this paper. The analysis for discrete-time systems is analogous and can be obtained using a similar procedure. Since the stability region of continuous-time systems is the open left-half complex plane, we have the following definition of the SR.
\begin{definition} {\bf(Sparse stability radius)} The stability radius of the continuous-time system \eqref{eq:dynamics} is given by
\begin{align} \label{eq:cont_stab_rad}
r_C &\triangleq \inf\: \{\lVert \Delta \rVert_F : \alpha(A(\Delta)) \geq 0, \:\Delta \in \mathbf{\Delta}_S \subset \real^{m\times p}\},
\end{align}
where $A(\Delta)= A+B\Delta C$ and $\mathbf{\Delta}_S$ denotes the set of sparse perturbations characterized by the structure matrix $S$ (see \eqref{eq:sparse_pert_set}).  \oprocend
\end{definition}

We make the following assumption regarding the stability of the nominal system \eqref{eq:dynamics}.

\noindent \textbf{A1:} The matrix $A$ is stable, i.e., $\alpha(A)<0$.

Assumption \textbf{A1} ensures that the SR is strictly greater than zero. Since the eigenvalues of $A(\Delta)$ are a continuous function of $\Delta$, the infimum in \eqref{eq:cont_stab_rad} is achieved on the imaginary axis of the complex plane \cite{DH-AJP:05}. Thus, we have
\begin{align} \label{eq:cont_stab_rad1}
r_C & = \min\: \{\lVert \Delta \rVert_F : \alpha(A(\Delta)) = 0, \:\Delta \in \mathbf{\Delta}_S \}.
\end{align}
This motivates the reformulation of the sparse SR problem as the following optimization problem:
\begin{align} \label{eq:opt_cost}
\textbf{SR:} \hspace{10pt}\underset{\Delta  \in \real^{m\times p}, \:x\in\mathbb{C}^{n}, \:\omega\in \real}{\min} \quad \frac{1}{2} &\:||\Delta||_F^2  \hspace{70pt}
\end{align}
\ajustspaceandequationnumber
\begin{subequations}
\begin{align}  \label{eq:eigv_assgn}
\text{s.t.} \quad  (A+B\Delta C)x &=j \omega x , \\ \label{eq:eigv_norm}
 x^{\conjtrans}x &= 1, \\ \label{eq:spar_const1} 
 S^{c} \circ  \Delta &= 0,
\end{align}
\end{subequations}
where the eigenvalue-eigenvector constraint \eqref{eq:eigv_assgn} is a reformulation of the spectral constraint in \eqref{eq:cont_stab_rad1} in terms of an eigenvector-eigenvalue pair $(x,j\omega)$. The normalization constraint \eqref{eq:eigv_norm} is added to ensure uniqueness of the eigenvector. The sparsity constraint \eqref{eq:spar_const1} is a reformulation of $\Delta \in \mathbf{\Delta}_S$ (c.f. \eqref{eq:sparse_pert_set}).

Several remarks are in order for the optimization problem {\bf SR}. First, the eigenvalue-eigenvector constraint \eqref{eq:eigv_assgn} is not convex. As a result, the optimization problem {\bf SR} is not convex, and it may have multiple local minima. This is a typical property of all $2$-norm/$F$-norm, complex/real SR problems, as well as most other minimum distance problems \cite{DK-MV:15}. 

%\begin{remark} (\textbf{Non-convexity of \textbf{SR}}) \label{rem:non-convex}
%\oprocend
%\end{remark}

Second, besides assigning an eigenvalue(s) on the imaginary axis, the equality constraint in \eqref{eq:cont_stab_rad1} also requires the remaining eigenvalues of $A(\Delta)$ to lie in the open left-half complex plane. However, we have omitted this constraint in \textbf{SR} because it will be inherently satisfied by the \emph{global} minimum of \textbf{SR} due to (i) Assumption \textbf{A1}, (ii) the continuity properties of the eigenvalues of $A(\Delta)$, and (iii) the definition of SR in \eqref{eq:cont_stab_rad1}. However, a \emph{local} minimum of \textbf{SR} need not satisfy this constraint necessarily. Thus, all the local minima $\hat{\Delta}$ of \textbf{SR} should be verified against the constraint $\alpha(A(\hat{\Delta})) = 0$ (see Section \ref{sec:simulations} for an example).

Third, it may be possible that \textbf{SR} is non-feasible and there does not exist any $\Delta$ that satisfies constraints \eqref{eq:eigv_assgn}-\eqref{eq:spar_const1}.\footnote{A trivial example is: $A=\left[\begin{smallmatrix} -1 & 2\\0 & -2\end{smallmatrix}\right], B = C = I_2$ and $S = \left[\begin{smallmatrix} 0 & 1\\0 & 0\end{smallmatrix}\right].$ In this case, since only $A_{12}$ is allowed to be perturbed, the eigenvalues of $A(\Delta)$ will always be $\{-1,-2\}$ and cannot lie on the imaginary axis.} Such non-feasible cases are universally robust in the sense that no perturbation with the given sparsity structure can make the system unstable, and thus $r_C = \infty$. To avoid such cases, we make the following assumption:\\
\noindent \textbf{A2:} \textbf{SR} is feasible, i.e., there exists at least one perturbation $\Delta$ that satisfies \eqref{eq:eigv_assgn}-\eqref{eq:spar_const1}.

Finally, since $A(\Delta)$ is real, its eigenvalues are symmetric with respect to the real axis. Hence, if $(\hat{\Delta},\hat{\omega},\hat{x})$ is a local minima of \eqref{eq:opt_cost}, then $(\hat{\Delta},-\hat{\omega},\hat{x}^{*})$ is also a local minima.

%\begin{remark} (\textbf{Choice of norm})
%We choose the Frobenius norm to measure the perturbation size since it explicitly measures the element-wise perturbations. Further, it also allows to analytically compute the derivatives of the cost \eqref{eq:opt_cost}, thus enabling the use of gradient based procedures (see section \ref{}).  \oprocend
%\end{remark}

\section{Solution to the optimization problem}  \label{sec:soln_opt_prob}
In this section we present the optimality conditions for the local solutions of the optimization problem {\bf SR}, and characterize an optimal perturbation. We use the theory of Lagrange multipliers for equality-constrained minimization to derive the optimality conditions. We begin with a remark on the formalism involving complex variables.

\begin{remark} {\bf (Complex variables)} \label{rem:complex_formalism}
The eigenvalue-eigen- vector constraint in \eqref{eq:eigv_assgn} is a complex-valued constraint and it also induces the following conjugate constraint:
\begin{align} \label{eq:eigv_assgn_conj}
(A+B\Delta C)x^{*}=-j\omega x^{*}. 
\end{align}
We use the formalism wherein a complex number and its conjugate are treated as independent variables \cite{AH-DG:07,DHB:83} and, thus, we treat $x$ and $x^{*}$ as independent variables.  \oprocend
\end{remark}

In the theory of equality-constrained optimization, the first-order optimality conditions are
meaningful only when the optimal point satisfies the following regularity
condition: the Jacobian of the constraints, defined by $J_b$, is full rank. This regularity condition is mild and usually satisfied
for most classes of problems \cite{DGL:84}.  Before presenting the main
result, we derive the Jacobian and state the regularity condition for \textbf{SR}. The derivation of the Jacobian requires the vectorization of the sparsity constraint \eqref{eq:spar_const1}. Let $\delta \triangleq \text{vec}(\Delta)\in\mathbb{R}^{mp}$ and let $n_s$ denote the number of non-trivial sparsity constraints (i.e. number of $1$'s in $S^{c}$). Then, \eqref{eq:spar_const1} can be vectorized as:
\begin{align} \label{eq:spar_const_vec}
\mathsf{S} \delta = 0,
\end{align}
where $\mathsf{S} \in \{0,1\}^{n_s \times mp}$ is a binary matrix given by $\mathsf{S} = [e_{s_1},e_{s_2},\cdots,e_{s_{n_s}}]^{\trans}$ with $\{s_1,\cdots, s_{n_s}\}=$ supp(vec($S^{c}$)) being the set of indices indicating the $1$'{s} in vec($S^{c}$) and $e_i$ being the $i^{\text{th}}$ standard basis vector of $\mathbb{R}^{mp}$.

\noindent Recalling Remark \ref{rem:complex_formalism}, let $z\triangleq [x^{\trans},x^{\conjtrans},\delta^{\trans},\omega]^{\trans}$ be the vector containing all the variables of optimization problem $\textbf{SR}$. 

\begin{lemma}{\bf (Jacobian of the constraints)}\label{lem:jacobian} The Jacobian of the equality constraints \eqref{eq:eigv_assgn},\eqref{eq:eigv_norm}, \eqref{eq:eigv_assgn_conj}, \eqref{eq:spar_const_vec} is given by 
\begin{align*}%\label{eq:jacobian}
 & J_b(z)\!=\! \begin{bmatrix}
A(\Delta) -j\omega I & 0 & (Cx)^{\trans}\!\!\otimes\! B & -jx \\
0 & A(\Delta) + j\omega I   & (Cx^{*})^{\trans}\!\!\otimes\! B & jx^{*}\\
x^{\conjtrans} & x^{\trans}  & 0 & 0\\
0 & 0 & \mathsf{S} & 0
\end{bmatrix}\!\!.
\end{align*}
\end{lemma}

\begin{proof}
We construct the Jacobian $J_b$ by taking the derivatives of the constraints \eqref{eq:eigv_assgn}, \eqref{eq:eigv_norm}, \eqref{eq:eigv_assgn_conj} and \eqref{eq:spar_const_vec} with respect to $z$. Using \ref{prop:vec2}, constraint \eqref{eq:eigv_assgn} can be written as
\begin{align}   \label{eq:eigv_assgn_vec}
(A -j\omega I)x + [(Cx)^{\trans}\otimes B]\delta = 0.
\end{align}
Differentiating \eqref{eq:eigv_assgn} w.r.t. $x,\omega$ and \eqref{eq:eigv_assgn_vec} w.r.t $\delta$ yields the first (block) row of $J_b$. Similar derivatives of the conjugate constraint \eqref{eq:eigv_assgn_conj} w.r.t. $z$ yields the second (block) row of $J_b$. Differentiating constraint \eqref{eq:eigv_norm} w.r.t. $x$ and $x^{*}$ yields the third (block) row of $J_b$. Finally, differentiating \eqref{eq:spar_const_vec} w.r.t. $z$ yields the last (block) row of $J_b$. \oprocend
\end{proof}

Next, we provide the local optimality conditions for the optimization problem \textbf{SR}.

\begin{theorem}{\bf (Optimality conditions)}\label{thm:opt_feedback}
Let $(\hat{\Delta}, \hat{x}, \hat{\omega})$ (equivalently $\hat{z} = [\hat{x}^{\trans},\hat{x}^{\conjtrans},\hat{\delta}^{\trans},\hat{\omega}]^{\trans}$) satisfy the constraints \eqref{eq:eigv_assgn}-\eqref{eq:spar_const1}. Then, $(\hat{\Delta}, \hat{x}, \hat{\omega})$ is a local minimum of the optimization problem \textbf{SR}
if and only if
\begin{subequations}
\begin{align} \label{eq:opt_cond1}
\hat{\Delta} = -S\circ \Big[B^{\trans} \text{Re}(\hat{l} &\hat{x}^{\trans}) C^{\trans} \Big], 
\intertext{where $\hat{x}$ and $\hat{l}$ are the right and left eigenvectors of $A(\hat{\Delta})$, respectively, and satisfy}
(A+B\hat{\Delta}C)\hat{x}&=j \hat{\omega}\hat{x},  \label{eq:right_eigvec} \\ \label{eq:left_eigvec}
(A+B\hat{\Delta}C)^{\trans} \hat{l} &=j \hat{\omega}\hat{l} \qquad \text{and,}\\   \label{eq:realness}
\text{Im}(\hat{l}^{\trans}\hat{x}) &= 0, \\ \label{eq:regularity}
J_b(\hat{z})\:\: \text{is full } & \text{rank,}   \\ \label{eq:pos_definite}
P(\hat{z}) \hat{D} P(\hat{z}) & > 0,
\end{align}
\end{subequations}
where $\hat{D}$ is the Hessian defined as
\begin{align} \label{eq:hessian}
\hat{D} \triangleq \begin{bmatrix} 
0 & 0 & \hat{\tilde{L}}^{\conjtrans} & j \hat{l}^{*} \\
0 & 0 & \hat{\tilde{L}}^{\trans} & -j \hat{l} \\
\hat{\tilde{L}} & \hat{\tilde{L}}^{*} & 2I & 0 \\
-j\hat{l}^{\trans} & j \hat{l}^{\conjtrans} & 0 & 0
\end{bmatrix}, 
\end{align}  
with $\hat{\tilde{L}} \triangleq  C \otimes (B^{\trans}\hat{l})$, and $P(z)$ is the projection matrix of $J_b(z)$ given by $P(z) = I-J_b^{+}(z)J_b(z)$.

\end{theorem}

\begin{proof}
We prove the result using the Lagrange multiplier method for equality-constrained minimization. Let $l\in\mathbb{C}^{n},l^{*}$, $h\in\mathbb{R}$ and $M\in\real^{m \times p}$ be the Lagrange multipliers associated with constraints \eqref{eq:eigv_assgn}, \eqref{eq:eigv_assgn_conj}, \eqref{eq:eigv_norm} and \eqref{eq:spar_const1}, respectively. The Lagrangian function for the optimization problem \textbf{SR} is given by
\begin{align*}
\mc{L} \overset{\ref{prop:frob}}{=} &  \frac{1}{2}\:\text{tr}(\Delta^{\trans} \Delta) + \frac{1}{2} l^{\trans}(A+B\Delta C-j\omega I)x  \\
+ & \frac{1}{2} l^{\conjtrans}(A+B\Delta C+j\omega I)x^{*}  + h(x^{\conjtrans}x-1) \\
&+ \underbrace{1_{m}^{\trans} [M\circ (S^{c}\circ \Delta)]1_p}_{= \text{tr}[(M\circ S^{c})^{\trans} \Delta]}.
\end{align*} 
Next, we derive the first-order necessary conditions for a stationary point of \textbf{SR}. Differentiating $\mc{L}$ w.r.t. $x$ and setting to $0$, we get
\begin{align} \label{eq:l_der_X}
\frac{d}{dx}\mc{L} \overset{\ref{prop:der}}{=} \frac{1}{2} (A+B\Delta C-j\omega I)^{\trans}l+ x^{*}h = 0.
\end{align}
Pre-multiplying \eqref{eq:l_der_X} by $x^{\trans}$, we get
\begin{align*}
&\frac{1}{2}x^{\trans} (A+B\Delta C-j\omega I)^{\trans}l+x^{\trans} x^{*}h = 0, \\ 
& \overset{\eqref{eq:eigv_assgn},\eqref{eq:eigv_norm}}{\Longrightarrow} \: h = 0, 
\end{align*}
and thus,  from \eqref{eq:l_der_X}, we get \eqref{eq:left_eigvec}. Equation \eqref{eq:right_eigvec} is a restatement of \eqref{eq:eigv_assgn} for the optimal $(\hat{\Delta},\hat{x},\hat{\omega})$.

Differentiating $\mc{L}$ w.r.t. $\Delta$ and setting to $0$, we get
\begin{align} \label{eq:l_der_Delta}
\frac{d}{d\Delta}\mc{L} \overset{\ref{prop:der}}{=} \Delta + \text{Re}(B^{\trans}lx^{\trans}C^{\trans})  + M\circ S^{c}= 0. 
\end{align}
Taking the Hadamard product of \eqref{eq:l_der_Delta} with $S^{c}$ and using \eqref{eq:spar_const1} and the fact $S^{c}\circ S^{c} = S^{c}$, we get
\begin{align} \label{eq:l_der_Delta1}
S^{c} \circ \text{Re}(B^{\trans}lx^{\trans}C^{\trans})  + M\circ S^{c}= 0. 
\end{align}
Replacing $M\circ S^{c}$ from \eqref{eq:l_der_Delta1} in \eqref{eq:l_der_Delta}, we get \eqref{eq:opt_cond1} since $B$ and $C$ are real.

Finally, differentiating $\mc{L}$ w.r.t. $\omega$ and setting to $0$, we get
\begin{align*}
\frac{d}{d\omega}\mc{L} = \frac{1}{2}jl^{\conjtrans} x^{*} - \frac{1}{2}jl^{\trans} x = \text{Im}(l^{\trans}x) = 0. 
\end{align*}
Equation \eqref{eq:regularity} is the necessary regularity condition and follows from Lemma \ref{lem:jacobian}.
 
Next, we derive the second-order sufficient condition for a local minimum by calculating the Hessian of $\mc{L}$ w.r.t. $z$. Taking the differential of $\mc{L}$ twice, we get
\begin{align*}
d^2 \mc{L} &= \text{tr}((d\Delta)^{\trans}d\Delta) + l^{\trans}Bd\Delta C dx - j (d\omega)  l^{\trans}dx \\
&+ l^{\conjtrans}Bd\Delta C  dx^{*}  +j(d\omega) l^{\conjtrans}  dx^{*} +2h (dx)^{\conjtrans} dx \\
 \overset{\ref{prop:frob},h=0}{=}  & (d\delta)^{\trans}d\delta + \text{vec}^{\trans}(C^{\trans}(d\Delta)^{\trans}B^{\trans}l)dx - j (d\omega)  l^{\trans}dx \\
 &+ \text{vec}^{\trans}(C^{\trans}(d\Delta)^{\trans}B^{\trans}l^{*})dx^{*} + j(d\omega) l^{\conjtrans}  dx^{*} \\
 \overset{\ref{prop:vec1},\ref{prop:kron}}{=} & (d\delta)^{\trans}d\delta+ (d\delta)^{\trans}\tilde{L}dx - j (d\omega)  l^{\trans}dx \\
 &+ (d\delta)^{\trans}\tilde{L}^{*}dx^{*} + j(d\omega) l^{\conjtrans}  dx^{*} \\
 = &\frac{1}{2} [(dx)^{\conjtrans},(dx)^{\trans}, (d\delta)^{\trans}, d\omega] D \begin{bmatrix} dx\\dx^{*}\\d\delta\\d\omega \end{bmatrix},
\end{align*}
where $D$ is the Hessian (c.f. \ref{prop:grad_Hess}) defined in \eqref{eq:hessian}.
The sufficient second-order optimality condition for the optimization problem requires the Hessian to be positive-definite in the kernel of the Jacobian at the optimal point \cite{DGL:84}. That is,  
$y^\trans D y > 0, \; \forall y: J_b(z)y = 0$. This condition is equivalent to $P^{\trans}(z) D P(z) > 0$, since $J_b(z)y = 0$ if and only if $y = P(z)s$ for a complex $s$ \cite{DGL:84}. Since the projection matrix $P(z)$ is symmetric, \eqref{eq:pos_definite} follows, and the proof is complete. \oprocend
\end{proof}

The local optimality conditions in Theorem \ref{thm:opt_feedback} reveal the inherent properties of an optimal perturbation and the stability radius. First, \eqref{eq:opt_cond1} presents the explicit relations between the optimal perturbation $\hat{\Delta}$ and left and right eigenvectors of the optimally perturbed matrix $A(\hat{\Delta})$. Second, \eqref{eq:realness} shows that the inner product of the left-conjugate and right eigenvectors of the optimal perturbation is always real. Third, notice that the optimal perturbation in \eqref{eq:opt_cond1} always satisfies the sparsity constraint \eqref{eq:spar_const1} (since $S\circ S^{c} = 0$). 

 %Interestingly, this condition is equivalent to the condition for the equilibrium point of the ODE in \cite{NG-MG-TM-MLO:17},\footnote{Due to notational difference, the left eigenvector $y$ in \cite{NG-MG-TM-MLO:17} is equal to $l^{*}$ in our case.} which corresponds to the locally rightmost point of a spectral value set. 
 
%Second, \eqref{eq:realness} shows that the inner product of the left-conjugate and right eigenvectors of the optimal perturbation is always real. This condition is equivalent to the RP-compatible property of the right and left eigenvectors in \cite{NG-MG-TM-MLO:17}. These observations show that the solutions of our optimization problem have underlying connections with the geometry of the spectral value sets.

The optimality condition \eqref{eq:opt_cond1} can also be re-written as $\hat{\Delta} = -S \circ [B^{\trans}\hat{L}\hat{X}^{\trans}C^{\trans}]$, where $\hat{L}\triangleq [\text{Re}(\hat{l}),-\text{Im}(\hat{l})]$ and $\hat{X}\triangleq [\text{Re}(\hat{x}),\text{Im}(\hat{x})]$. This shows that, although \\ \noindent rank$(B^{\trans}\hat{L}\hat{X}^{\trans}C^{\trans})\leq 2$, the rank of $\hat{\Delta}$ can be greater than $2$. In contrast, the optimal perturbation for real, \emph{non-sparse}, $2$-norm/$F$-norm SR always has rank less that or equal to $2$ \cite{LQ-BB-AR-EJD-PMY-JCD:95}.

%This shows that if $\min\{m,p\}\geq 2$, then $\hat{\Delta}$ is generically rank-$2$, otherwise, it is generically rank-$1$.\footnote{The case $\hat{\omega}=0$ is an exception, wherein $\hat{\Delta} = -B^{\trans} \hat{l} \hat{x}^{\trans} C^{\trans}$ is rank-$1$.} Note that the optimal perturbation for real, structured, $2$-norm SR also exhibits similar rank properties \cite{LQ-BB-AR-EJD-PMY-JCD:95}.

\begin{remark} {\bf (Eigenvector normalization constraints)} \label{rem:norm_const}
From the proof of Theorem \ref{thm:opt_feedback}, we observe that the Lagrange multiplier $h$ associated with the eigenvector normalization constraint \eqref{eq:eigv_norm} is zero. This implies that these constraints are 
redundant. Further, observe from \eqref{eq:opt_cond1} that any pair of left and right eigenvectors $\{\hat{l}/\beta,\hat{x}\beta\}$ with $\beta \in \mathbb{C}\backslash\{0\}$, will result in the same perturbation matrix $\hat{\Delta}$. 
We can always choose a suitable $\beta$ to normalize $\hat{x}$, and thus, we ignore these eigenvector normalization constraints in the remainder of the paper. \oprocend
\end{remark}

The solution of the optimization problem \textbf{SR} can be obtained by numerically/iteratively solving the optimality equations \eqref{eq:opt_cond1}-\eqref{eq:realness} using any non-linear equation solving technique. The regularity and local minimum property of the solution can be verified using \eqref{eq:regularity} and \eqref{eq:pos_definite}, respectively. Finally, the local minima should be verified against $\alpha(A(\hat{\Delta})) = 0$. Since the optimization problem is not convex, only local minima can be obtained via this procedure. To improve upon the local solutions and to capture the global minimum,  the above procedure can be repeated for several different initial conditions. Clearly, finding the global minimum is not guaranteed in all cases. In this case, the procedure provides an upper bound to the SR.  %Note that the existence of local minimas is an inherent property of the most of the minimum distance problems \cite{} and is also present for the case of computing complex or real stability radius w.r.t. 2-norm \cite{}.

In this paper, instead of solving the optimality equations, we use a penalty based approach using gradient descent to obtain the local solutions. Details of this approach and the corresponding algorithm are provided in the next section.

\section{Gradient based solution algorithm}   \label{sec:grad_algo}
In this section, we present an iterative gradient based algorithm to obtain a local solution to the optimization problem \textbf{SR}. We use the penalty based optimization approach and the Sylvester equation based parametrization to convert the constrained optimization problem \eqref{eq:opt_cost} into an unconstrained optimization problem. Note that we ignore the eigenvector normalization constraints \eqref{eq:eigv_norm} (c.f. Remark \ref{rem:norm_const}) hereafter.

We begin by reformulating \eqref{eq:eigv_assgn} as a purely real constraint. Let $X\triangleq [\text{Re}(x),\text{Im}(x)]\in \real^{n \times 2}$. Then, \eqref{eq:eigv_assgn} is equivalent to
\begin{align} \label{eq:eigv_assgn_real}
(A+B\Delta C)X = \omega X \bar{I},
\end{align}
where $\bar{I} \triangleq \left[ \begin{smallmatrix} 0 & 1\\ -1 & 0 \end{smallmatrix}\right ]$. Next, we use the Sylvester equation based parametrization \cite{SPB-EDS:82} and define $G\triangleq \Delta CX \in \real^{m \times 2}$. It follows that \eqref{eq:eigv_assgn_real} is equivalent to
\begin{subequations}
\begin{align} \label{eq:eigv_assgn_syl_1}
AX-\omega X\bar{I} &= -BG, \\ \label{eq:eigv_assgn_syl_2}
G&= \Delta CX.
\end{align}
\end{subequations}
Note that \eqref{eq:eigv_assgn_syl_1} is a Sylvester equation in $X$. Due to Assumption \textbf{A1}, $\Gamma(A)\cap \Gamma(-\omega \bar{I})=\emptyset$ and, thus, \eqref{eq:eigv_assgn_syl_1} has a unique solution for any given $(G,\omega)$. Further, for any $G$,  \eqref{eq:eigv_assgn_syl_2} has a solution if $CX\in \real^{p\times 2}$ is rank two. Thus, we make the following assumption.

\noindent \textbf{A3:} For a given $(G,\omega)$, $CX$ is full column rank,\footnote{This requires $p\geq 2$.} where $X$ is the unique solution of \eqref{eq:eigv_assgn_syl_1}. 

Under Assumption $\textbf{A3}$, we can solve \eqref{eq:eigv_assgn_syl_2} to obtain $\Delta$, which, in general, may not be unique. Since we wish to minimize the norm of the perturbation, we choose the unique minimum norm solution of  \eqref{eq:eigv_assgn_syl_2}, which is given by $\Delta = G(CX)^{+}$. To summarize, using the Sylvester equation based parametrization, we can freely vary $(G,\omega)$ (under Assumption \textbf{A3}) and compute the corresponding $X$ using \eqref{eq:eigv_assgn_syl_1} and $\Delta = G(CX)^{+}$. 

Next, we use the penalty based optimization approach \cite{DGL:84} and modify the cost function to include a penalty when the sparsity constraints are violated. The penalty is imposed by weighing individual entries of the perturbation using a weighing matrix $W\in \real^ {m \times p}$ given by
\begin{align*}
W_{ij} = 
\begin{cases}
1  \quad &\text{if} \:\:\  S_{ij} = 1, \:\text{and} \\
{\bf w} \gg 1 \quad &\text{if} \:\:\ S_{ij} = 0.
\end{cases}
\end{align*}
Using the weighing matrix $W$, the modified cost becomes 
\begin{align*}
J_{W} = \frac{1}{2}\:||W \circ \Delta||_F^2 \overset{(a)}{=} \frac{1}{2}\: ||\Delta||_F^2 + \frac{1}{2} ({\bf w}^2-1) ||S^{c}\circ \Delta||_F^2,
\end{align*}
where $(a)$ follows from (i) $W = 1_{m\times p} + ({\bf w}-1) S^{c}$ and (ii) $\text{tr}(\Delta^{\trans} (S^{c} \circ \Delta))=\text{tr}((S^{c} \circ\Delta)^{\trans} (S^{c} \circ \Delta))$. Using the penalized cost and Sylvester equation based parametrization, the constrained optimization problem \eqref{eq:opt_cost} can be reformulated as as the following unconstrained optimization problem in variables $G,\omega$:
\begin{align} \label{eq:opt_cost_repar}
 \hspace{30pt}\underset{G  \in \real^{m\times 2}, \:\omega\in \real}{\min} \quad J_W= \frac{1}{2}\:||W\circ \Delta||_F^2  \hspace{30pt}
\end{align}
\ajustspaceandequationnumber
\begin{subequations}
\begin{align}  \label{eq:eigv_assgn_repar}
\text{s.t.} \quad  AX-\omega X\bar{I} &= -BG, \\ \label{eq:fg}
 \Delta &=G(CX)^{+}.
\end{align}
\end{subequations}
%Note that as the weight $w$ corresponding to the sparse entries of $\Delta$ increases, the solution of \eqref{eq:opt_cost_repar} becomes more sparse.

We aim to solve the unconstrained problem \eqref{eq:opt_cost_repar} using a gradient descent approach. The next result provides analytical expressions for the gradient and Hessian of the cost in \eqref{eq:opt_cost_repar}. Let $g\triangleq \text{vec}(G)\in\real^{2m},x_v = \text{vec}(X)\in\real^{2n}$, and let the free variables of \eqref{eq:opt_cost_repar} be denoted by $\bar{z} \triangleq [g^{\trans},\omega]^{\trans}$.

\begin{lemma}  {(\bf Gradient and Hessian)} \label{lem:grad_Hess}
The gradient and Hessian of the cost $J_W$ in \eqref{eq:opt_cost_repar} are given by
\begin{align} \label{eq:grad}
&\frac{dJ_W}{d\bar{z}} = \underbrace{\begin{bmatrix}\tilde{X}^{+}(I+\tilde{\Delta} \tilde{A}(\omega)^{-1} \tilde{B}) & \tilde{X}^{+}\tilde{\Delta} \tilde{A}(\omega)^{-1} \tilde{I} x_v \end{bmatrix}^{\trans}}_{\textstyle \triangleq Z(\Delta,X,\omega)} \bar{W} \: \delta,\\ \label{eq:Hess}
&\frac{d^{2}J}{d^{2}\bar{z}} \triangleq H(\Delta,X,\omega) = Z\bar{W}Z^{\trans}\! + M\! + M^{\trans}\!, \quad \text{where} \\
&\tilde{X} = ((CX)^{\trans}\otimes I), \quad \tilde{B} = I \otimes B, \quad \tilde{\Delta} = I\otimes (\Delta C), \nonumber \\
&\tilde{I} = \bar{I}\otimes I,\quad \tilde{A}(\omega) = I \otimes A + \omega(\bar{I}\otimes I), \nonumber \\
&e_{2m+1}\! = \![0,\cdots,0,1]^{\trans}\!\in \mathbb{R}^{2m+1},  \bar{W} \triangleq \text{diag}(\text{vec}(W\!\circ\! W)), \nonumber \\
M &\!=\! \begin{bmatrix}\tilde{B} & \tilde{I}x_v \end{bmatrix}^{\trans} \!\!\!\tilde{A}(\omega)^{-\trans}\! \Big[((CX)^{+} (W\!\circ \!W\!\circ\!\Delta)^{\trans}\! \otimes\! C^{\trans} ) T_{m,p} Z^{\trans} \nonumber \\ \nonumber
&\qquad\qquad\qquad\qquad\quad - \tilde{I}^{\trans} \tilde{A}(\omega)^{-\trans} \tilde{\Delta}^{\trans} (\tilde{X}^{+})^{\trans}\bar{W}\delta e^{\trans}_{2m+1}\Big].
\end{align}
\end{lemma}

\begin{proof}
Taking differential of \eqref{eq:eigv_assgn_repar}, we get 
\begin{align} \label{eq:diff_XG}
AdX - \omega dX \bar{I} - d\omega X \bar{I} = - BdG.
\end{align}
Vectorizing \eqref{eq:diff_XG} using \ref{prop:vec1}, and using $\bar{I}^{\trans} = -\bar{I}$, we get
\begin{align}
(I\otimes A)dx_v &+  \omega(\bar{I}\otimes I)dx_v + (\bar{I}\otimes I)x_vd\omega = -(I\otimes B)dg  \nonumber \\ 
\Rightarrow dx_v &= -\tilde{A}^{-1}(\omega)(\tilde{B}dg+(\bar{I}\otimes I)x_vd\omega) \nonumber \\ \label{eq:dx}
& = - \underbrace{\tilde{A}(\omega)^{-1} \begin{bmatrix}\tilde{B} & \tilde{I}x_v \end{bmatrix}}_{\textstyle \triangleq Y}  \underbrace{\left[\begin{smallmatrix} dg\\d\omega \end{smallmatrix}\right]}_{\textstyle d\bar{z}}.
\end{align}
Taking the differential of \eqref{eq:eigv_assgn_syl_2} and vectorizing yields
\begin{align}
d\Delta  CX& + \Delta C dX = dG \label{eq:dG}\\
\overset{\ref{prop:vec1}}{\Rightarrow} \tilde{X} d\delta &= dg-\tilde{\Delta} dx_v \nonumber \\ \label{eq:d_delta}
\Rightarrow  d\delta &= \tilde{X}^{+}(dg-\tilde{\Delta} dx_v)  \overset{\eqref{eq:dx}}{=} Z^{\trans}(\Delta,X,\omega) d\bar{z}. 
\end{align}
%Taking the differential of \eqref{eq:fg} under assumption $\textbf{A2}$, we get
%\begin{align} \label{eq:d_Delta}
%d\Delta &\overset{\ref{prop:diff_pinv}}{=} dG (CX)^{+} - G(CX)^{+} C dX (CX)^{+} \nonumber \\
%&-\underbrace{G [I\!-\!(CX)^{+}CX]}_{=0, \text{ since } G=\Delta C X }(CdX)^{\trans} ((CX)^{+})^{\trans} (CX)^{+}.
%\end{align}
%Vectorizing \eqref{eq:d_Delta} and using \eqref{eq:fg}, we get
%\begin{align}
%d\delta &\overset{\ref{prop:vec1},\ref{prop:vec2}}{=} [((CX)^{+})^{\trans} \otimes I]dg - [((CX)^{+})^{\trans} \otimes (\Delta C)]dx_v \nonumber \\
%& \overset{\ref{prop:kron1}}{=}  \tilde{X}^{+}dg - \tilde{X}^{+} \tilde{\Delta} dx_v \nonumber \\ \label{eq:d_delta}
%& \overset{\eqref{eq:dx}}{=} Z^{\trans}(\Delta,X,\omega) d\bar{z}.
%\end{align}
Now $J_{W} \overset{\ref{prop:frob},\ref{prop:had2}}{=} \frac{1}{2} (\text{vec}(W)\circ \delta)^{\trans}(\text{vec}(W)\circ \delta) =\frac{1}{2} \delta^{\trans}\bar{W}\delta$. Thus, $dJ_{W} = \delta^{\trans}\bar{W}(d\delta)$. Using \eqref{eq:d_delta} and \ref{prop:grad_Hess}, we get the gradient in \eqref{eq:grad}.

Next we derive the Hessian of $J_W$. Taking the differential of $dJ_{W} = \delta^{\trans}\bar{W}(d\delta)$, we get
\begin{align} \label{eq:dd_J}
 d^{2}J_W = (d\delta)^{\trans}\bar{W}(d\delta) + (d^{2}\delta)^{\trans} \bar{W}\delta.
\end{align}
Since $G$ and $\omega$ are free variables, their second order differentials $d^{2}G$ and $d^{2}\omega$ are zero \cite{JRM-HN:99}. Taking the differential of \eqref{eq:diff_XG} and vectorizing, we have
\begin{align}
&Ad^{2}X - \omega d^{2}X \bar{I} - 2d\omega dX \bar{I} =0 \nonumber \\ \label{eq:dd_x}
&\overset{\ref{prop:vec1}}{\Rightarrow} d^{2} x_v = -2d\omega \tilde{A}^{-1}(\omega) \tilde{I} dx_v\! \overset{\eqref{eq:dx}}{=} \!2 d\omega \tilde{A}^{-1}(\omega) \tilde{I} Y d\bar{z}.
\end{align}
Taking the differential of \eqref{eq:dG} and vectorizing yields
\begin{align}
&d^{2}\Delta  CX + \Delta C d^{2}X + 2d\Delta C dX = 0 \nonumber \\
&\overset{\ref{prop:vec1}}{\Rightarrow} \tilde{X} d^{2}\delta + \tilde{\Delta} d^{2}x_v + 2 (I \otimes d\Delta C) dx_v = 0 \nonumber \\
&\overset{\eqref{eq:dx},\eqref{eq:dd_x}}{\Rightarrow} d^{2}\delta= 2\tilde{X}^{+} [(I \otimes d\Delta C) - d\omega \tilde{\Delta} \tilde{A}(\omega)^{-1} \tilde{I}] Y d\bar{z} \nonumber \\
& \Rightarrow (d^{2}\delta)^{\trans} \bar{W}\delta = 2d\bar{z}^{\trans} Y^{\trans} (I \otimes (d\Delta C)^{\trans}) (\tilde{X}^{+})^{\trans}\bar{W}\delta \nonumber \\ \label{eq:dd_delta_t_delta}
&\hspace*{50pt}- 2d\bar{z}^{\trans} Y^{\trans}  \tilde{I}^{\trans} \tilde{A}(\omega)^{-\trans} \tilde{\Delta}^{\trans} (\tilde{X}^{+})^{\trans}\bar{W}\delta d\omega.
\end{align}
From the first term on right side of \eqref{eq:dd_delta_t_delta}, we have
\begin{align}
&(I \otimes (d\Delta C)^{\trans}) (\tilde{X}^{+})^{\trans} \bar{W}\delta  \nonumber \\ 
&\overset{\ref{prop:kron},\ref{prop:kron1}}{=} (I \otimes (d\Delta C)^{\trans}) ((CX)^{+} \otimes I) \bar{W}\delta \nonumber \\
& \overset{\ref{prop:kron1}}{=} ((CX)^{+} \otimes (d\Delta C)^{\trans}) \bar{W}\delta \nonumber \\ 
&\overset{(a)}{=} \text{vec}(C^{\trans} (d\Delta)^{\trans} (W\circ W \circ\Delta) ((CX)^{+})^{\trans}) \nonumber \\
&\overset{\ref{prop:vec2},\ref{prop:per_mat}}{=} ((CX)^{+} (W\circ W\circ \Delta)^{\trans} \otimes C^{\trans} ) T_{m,p} d\delta \nonumber \\ \label{eq:dd_delta_t_delta_p1}
& \overset{\eqref{eq:d_delta}}{=}\! ((CX)^{+} (W\!\circ\! W\!\circ\! \Delta)^{\trans} \!\otimes\! C^{\trans} ) T_{m,p} Z^{\trans}(\Delta,X,\omega) d\bar{z},
\end{align}
where $(a)$ follows from \ref{prop:vec2} and the relation $\text{vec}(W\circ W \circ \Delta) = \bar{W} \delta$.
Substituting \eqref{eq:dd_delta_t_delta_p1} and $d\omega = e_{2m+1}^{\trans} d\bar{z}$ in \eqref{eq:dd_delta_t_delta}, we get
\begin{align} \label{eq:dd_delta_t_delta1}
(d^{2}\delta)^{\trans} \bar{W}\delta = 2d\bar{z}^{\trans} M d\bar{z} = d\bar{z}^{\trans} (M + M^{\trans}) d\bar{z}. 
\end{align}
Substituting \eqref{eq:dd_delta_t_delta1} and \eqref{eq:d_delta} in \eqref{eq:dd_J}, we get
\begin{align} \label{eq:dd_J1}
d^{2}J_W = d\bar{z}^{\trans} (Z \bar{W} Z^{\trans} +M + M^{\trans})d\bar{z}.
\end{align}
Using \eqref{eq:dd_J1} and \ref{prop:grad_Hess}, we get the Hessian in \eqref{eq:Hess}. \oprocend
\end{proof}

Using Lemma \ref{lem:grad_Hess}, we present a gradient/Newton descent Algorithm \ref{algo:grad_New_des} to solve the optimization problem \eqref{eq:opt_cost_repar}. 
\begin{algorithm} \label{algo:grad_New_des}
  \KwIn{$A,B,C,W,g_0,\omega_0.$}
  \KwOut{Local minima $(\Delta,X,\omega)$ of \eqref{eq:opt_cost_repar}.}
  \BlankLine
  \textbf{Initialize:} $\bar{z}_0\!=\!\begin{bmatrix} g_0\\ \omega_0 \end{bmatrix}$,$x_0\leftarrow -\tilde{A}(\omega_0)^{-1}\tilde{B}g_0 ,\delta_0  \leftarrow \tilde{X}_0^{+}g_0$ \\
 \Repeat{\textup{convergence}}
 {\nl $\beta \leftarrow$ Update step size (see below)\label{line:step_size}\;
 \nl  $\bar{z} \leftarrow  \bar{z} - \beta Z(\Delta,X,\omega)\bar{W}\delta\quad$ \textbf{or} \label{line:grad_des}\; 
 \nl $\bar{z} \leftarrow  \bar{z}- \beta [H(\Delta,X,\omega)+V]^{-1} \!Z(\Delta,X,\omega)\bar{W}\delta$\label{line:newton_des}\;
 \vspace{4pt}
 \nl $x \leftarrow -\tilde{A}(\omega)^{-1}\tilde{B}g$ \label{line:x_eval}\;
 \nl $\delta \leftarrow \tilde{X}^{+}g$ \label{line: delta_eval}\;
 }
 \Return{$(\Delta,X,\omega)$}
  \caption{Gradient/Newton descent for \textbf{SR}}
\end{algorithm}
Steps \ref{line:grad_des} and $\ref{line:newton_des}$ of Algorithm \ref{algo:grad_New_des}  represent gradient and damped Newton descent steps, respectively. In the Newton descent step, the Hessian $H(\Delta,X,\omega)$ is required to be positive-definite. To satisfy this property, we add the term $V = \epsilon I - M -M^{\trans}$ to the Hessian with $0<\epsilon \ll 1$ \cite{DGL:84}. Further, the step size $\beta$ can be updated using backtracking line search or Armijo's rule \cite{DGL:84}. In general, the Newton descent converges faster as compared to gradient descent. For a detailed discussion of the two algorithms, the interested reader is referred to \cite{DGL:84}. Further, steps \ref{line:x_eval} and \ref{line: delta_eval} are obtained by vectorizing \eqref{eq:eigv_assgn_repar} and \eqref{eq:fg}, respectively. The computational effort in each iteration of Algorithm \ref{algo:grad_New_des} mainly results from computing the pseudoinverse of $\tilde{X}$, and the inverses of $\tilde{A}(\omega)$ and  $H(\Delta,X,\omega)+V$. Finally, if Assumption \textbf{A3} is not satisfied in any iteration of Algorithm \ref{algo:grad_New_des}, (i.e., $CX$ is not full column rank), then we can slightly modify $(g,\omega)$ to ensure that \textbf{A3} is satisfied and continue the iterations.

\begin{remark} {\bf(Optimality at $\mathbf{\hat{\omega} = 0}$)} Algorithm \ref{algo:grad_New_des} is initialized at some $\omega_0$ and it updates $\omega$ at each step, eventually converging to a locally optimal $\hat{\omega}$. The perturbation $\Delta$ computed at each iteration of algorithm assigns two eigenvalues of $A(\Delta)$ at $\pm j\omega$. As a consequence, in cases where $\hat{\omega} = 0$ is a local minima of \eqref{eq:opt_cost_repar}, the algorithm converges to a perturbation $\hat{\Delta}$ such that $A(\hat{\Delta})$ has eigenvalue $0$ with multiplicity two. Clearly, this is not the optimal solution since instability results from at least one (and not necessarily two) eigenvalue of $A(\hat{\Delta})$ being at the origin. To address this special case ($\hat{\omega} = 0$), we can use the following eigenvalue assignment equation instead of \eqref{eq:eigv_assgn_real}
\begin{align*}
(A+B\Delta C)x = 0,
\end{align*} 
where $x\in\real^{n}$, and develop an algorithm (analogous to Algorithm \ref{algo:grad_New_des}) using a similar Sylvester equation based
 parame-trization method. \oprocend
\end{remark}

\begin{remark}{ \bf(Choice of weights)}
As the weight ${\bf w}$ increases, an optimal solution of \eqref{eq:opt_cost_repar} satisfies the sparsity constraints \eqref{eq:spar_const1} with increasing accuracy. However, an increase in the weights also reduces the convergence speed of Algorithm \ref{algo:grad_New_des}. Thus, there exists a trade-off between the accuracy of the sparse solutions and the convergence time of the algorithm. Convergence theory of gradient/Newton descent and the penalty based method is well established \cite{DGL:84}, and therefore, is not the primary focus of this paper. \oprocend
\end{remark}

\section{Simulation Studies} \label{sec:simulations}
In this section, we present numerical simulation studies of our algorithm. To begin, we consider the following example from \cite{LQ-BB-AR-EJD-PMY-JCD:95}:
\begin{align*}
 \small A &=\begin{bmatrix}
79 &   20 &  -30 &  -20\\
-41 &  -12 &   17 &   13\\
 167 &   40 &  -60 &  -38\\
 33.5 &    9 &  -14.5 & -11
\end{bmatrix}\hspace*{-3pt},
B =\! \begin{bmatrix}
0.2190  &   0.9347\\
0.0470  &  0.3835\\
0.6789  &  0.5194\\
0.6793  &  0.8310\\
\end{bmatrix}\!\!, \\
C &= \begin{bmatrix}
0.0346  &  0.5297 &   0.0077  &  0.0668\\
0.0535  &  0.6711  &  0.3848  &  0.4175
\end{bmatrix}. \normalfont
\end{align*}
The eigenvalues of $A$ are \{$-1\pm j , -1 \pm 10j$\}. We consider two cases:\\
\noindent \emph{Case 1:} No sparsity constraints, i.e., $S = \left[\begin{smallmatrix} 1 & 1\\1 & 1\end{smallmatrix}\right]$,\\
 \noindent \emph{Case 2:} Only the diagonal entries of $\Delta$ are allowed to be perturbed, i.e., $S = \left[\begin{smallmatrix} 1 & 0\\0 & 1\end{smallmatrix}\right]$. 
 
 The weight in the weighing matrix $W$ is chosen as ${\bf w}=100$. Table \ref{tab:local_min} shows the local minima of optimization problem \eqref{eq:opt_cost} obtained by Algorithm \ref{algo:grad_New_des} for both the cases. The first local minimum is also the global minimum for both the cases. Note that the second local minimum for case 1 satisfies $\alpha(A(\hat{\Delta}^{(2)})) = 0$, whereas the second local minimum for case 2 does not satisfy this constraint (c.f. discussion after \eqref{eq:spar_const1}). Thus, it is not a valid local minimum.

Next, we illustrate the relation between the local minima of our optimization problem and the geometry of the spectral value sets. Spectral value set captures the region in which all possible eigenvalues of the perturbed system can lie, and for $\eta\geq 0$, is defined as:
\begin{align*}
\mathcal{S}_{\eta} \triangleq \{\Gamma(A+B\Delta C): ||\Delta||_F\leq \eta, \Delta\in\mathbf{\Delta}_S\}.
\end{align*}
Figure \ref{fig:spec_sets} shows the spectral value sets in the complex plane corresponding to the local minima in Table \ref{tab:local_min}. We observe that the local minima are precisely the cases when the \emph{locally} right-most points of the spectral value sets\footnote{$\mathcal{S}_{\eta}$ for this example was visualized by performing an exhaustive search over $\Delta = \left[ \begin{smallmatrix} \Delta_{11} & 0 \\ 0 & \Delta_{22}\end{smallmatrix}\right]$ such that $\Delta_{11}^2 +  \Delta_{22}^2 \leq \eta^2,$ and plotting $\Gamma(A(\Delta))$.} intersect with the imaginary axis. 

%This clearly illustrates the equivalence between our optimization based framework and the spectral set abscissa based framework of \cite{LQ-BB-AR-EJD-PMY-JCD:95}.

%Table \ref{tab:2F_SR_comp} presents a comparison between the $2$-norm SR ($\text{SR}_2, \hat{\Delta}_2,\hat{\omega}_2$) computed using the characterization in \cite{LQ-BB-AR-EJD-PMY-JCD:95},\footnote{The numerical values reported in \cite{LQ-BB-AR-EJD-PMY-JCD:95} are slightly inaccurate. We recompute these values to improve their accuracy.} and the $F$-norm SR ($\text{SR}_F, \hat{\Delta}_F,\hat{\omega}_F$) computed using Algorithm \ref{algo:grad_New_des}. Both SR problems exhibit two local minima (counting only those for which $\hat{\omega} \geq 0$; c.f. discussion below Remark \ref{rem:non-convex}) as shown in Table \ref{tab:2F_SR_comp}, and the first local minimum is also the global minimum. For the global minimum, observe that the optimal values of $\hat{\omega}$ are nearly equal for the $2$-norm and $F$-norm case. Further, the singular values of $\hat{\Delta}_F^{(1)}$ are $0.5159$ and $0.0028$, indicating that it is very close to a rank-$1$ matrix. As a result, the $F$-norm SR is very close to the $2$-norm SR.

\begin{table}[h!]
\caption{Minima obtained via Algorithm \ref{algo:grad_New_des}}
\begin{center}
\setlength{\tabcolsep}{2pt}
  \begin{tabular}{| l | l | }
    \hline 
&\\[-0.7em]

\multicolumn{1}{|c|}{Case $1$}  & \multicolumn{1}{c|}{Case $2$}   \\ \hline
\rule{0pt}{15pt}$\hat{\Delta}_1= 
\begingroup % keep the change local
\setlength\arraycolsep{3pt}
 \begin{bmatrix}
-0.0332 &  -0.0717 \\
   \phantom{-} 0.1975 &   \phantom{-}0.4700
\end{bmatrix} \endgroup $
& 
$\hat{\Delta}_1 =
\begingroup % keep the change local
\setlength\arraycolsep{3pt}
 \begin{bmatrix}
-0.0418 &  0.0000 \\
   0.0000 &   0.5638
\end{bmatrix} \endgroup $  \\

\rule{0pt}{12pt}\hspace{-3pt} $r_C = ||\hat{\Delta}_1||_F = 0.5159$ & $r_C=|| \hat{\Delta}_1||_F = 0.5653$  \\
\rule{0pt}{12pt}$\hat{\omega}_1 =1.3753$ &  $\hat{\omega}_1= 1.3365$ \\ 
 \rule{0pt}{21pt} $\hat{x}_1= \begin{bmatrix} \phantom{-}0.1340 - 0.0022j \\\phantom{-} 0.3692 + 0.0456j \\\phantom{-} 0.7733 + 0.2579j \\-0.2504 - 0.3411j \end{bmatrix}$   &  $\hat{x}_1= \begin{bmatrix} 0.0905 - 0.0971j \\  0.2152 - 0.3108j \\  0.3295 - 0.7459j  \\ 0.0799 + 0.4099j \end{bmatrix}$\\
 \rule{0pt}{21pt} $\hat{l}_1 =\: \begin{bmatrix} -1.3796 + 0.4056j \\ -0.5825 - 0.1855j \\ \phantom{-}0.4576 - 0.1326j \\ \phantom{-}0.2771 - 0.4659j  \end{bmatrix}$ & $\hat{l}_1 =\: \begin{bmatrix} -0.7660 - 1.5362j \\  -0.6177 - 0.3611j  \\ \phantom{-}0.2590 + 0.5098j \\ -0.1785 + 0.6099j \end{bmatrix}$\\
\hline
\rule{0pt}{15pt}$\hat{\Delta}_2= 
\begingroup % keep the change local
\setlength\arraycolsep{3pt}
\begin{bmatrix}
\phantom{-}0.1841  & \phantom{-}0.5173\\
   -0.8050  & -0.4151
\end{bmatrix}\endgroup$ & 
$\hat{\Delta}_2 = 
\begingroup % keep the change local
\setlength\arraycolsep{3pt}
\begin{bmatrix}
4.8818  & 0.0000\\
   0.0000  &-0.8898
\end{bmatrix}\endgroup$  \\ 

\rule{0pt}{12pt}$||\hat{\Delta}_2||_F = 1.0592$ & $||\hat{\Delta}_2||_F = 4.9622$  \\ 
\rule{0pt}{12pt}$\hat{\omega}_2 = 10.8758$ &  $\hat{\omega}_2 = 11.0790$ \\ 
\rule{0pt}{21pt}$\hat{x}_2 = \begin{bmatrix} \phantom{-}0.2032 + 0.3252j \\ \phantom{-}0.0505 - 0.2331j \\ \phantom{-}0.7184 + 0.4695j \\ -0.0532 + 0.2381j \end{bmatrix}$ & $\hat{x}_2 = \begin{bmatrix} -0.1927 - 0.3320j \\  \phantom{-}0.2084 + 0.0795j   \\   -0.0770 - 0.8566j \\   -0.2483 - 0.0397j   \end{bmatrix}$\\
\rule{0pt}{21pt}$\hat{l}_2 =\: \begin{bmatrix} \phantom{-}8.4752 + 9.7446j \\ \phantom{-}1.8464 + 2.4744j\\ -3.6880 - 3.0693j \\ -2.4960 - 1.8877j \end{bmatrix}$ & $\hat{l}_2 =\: \begin{bmatrix} \phantom{-}435.71 + 76.99j \\112.50 - 0.98j \\ -153.99 - 56.91j \\ \phantom{-} \phantom{-}95.93 - 42.90j \end{bmatrix}$\\
\hline 
 
\end{tabular}

\end{center}
\label{tab:local_min}
\end{table}

\begin{figure}[h!]
  \centering
  \subfigure[$\eta = 0.5159$]{
  \includegraphics[width=.45\columnwidth]{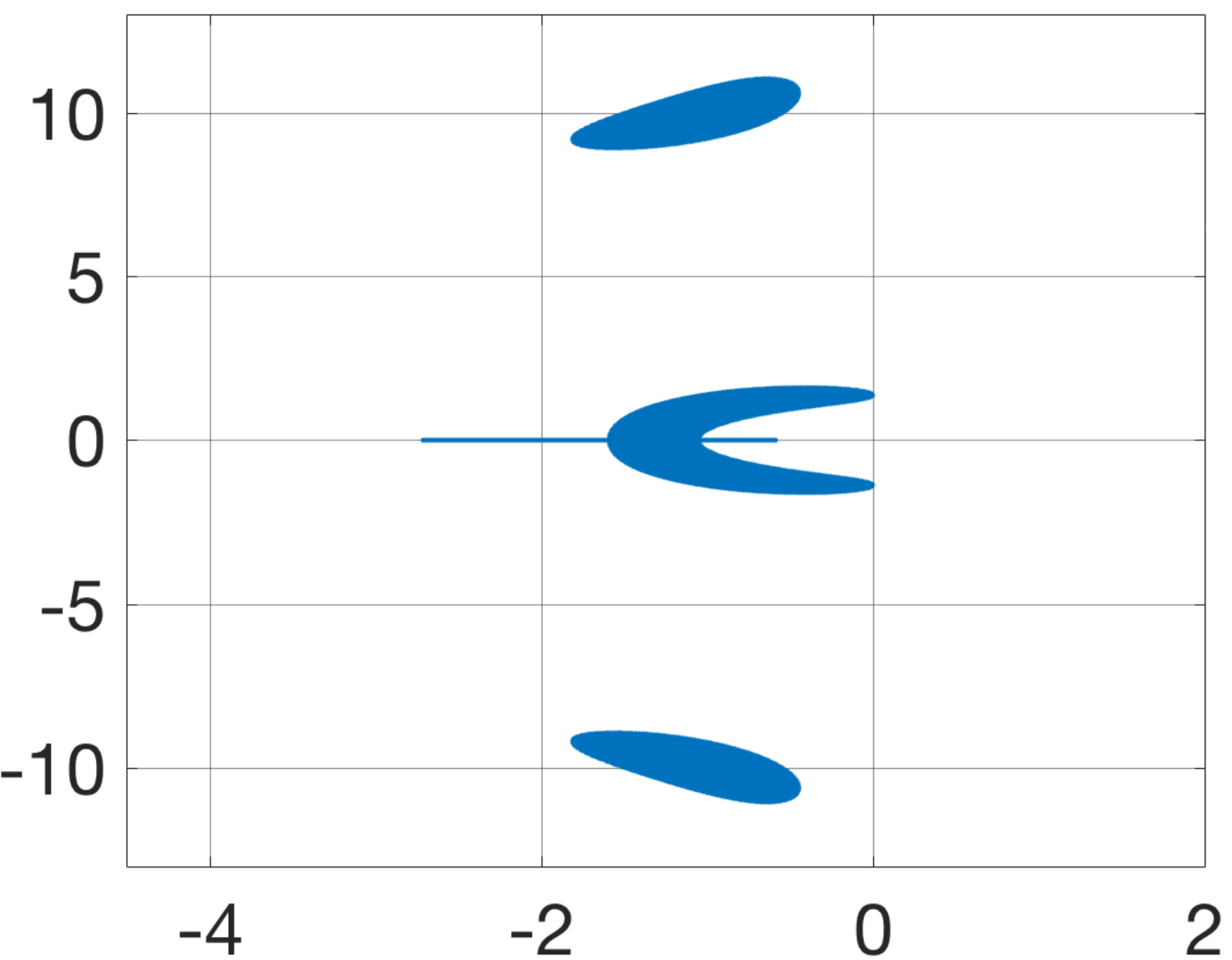} \label{fig:ss1} } 
  \subfigure[$\eta = 1.0592$]{
  \includegraphics[width=.45\columnwidth]{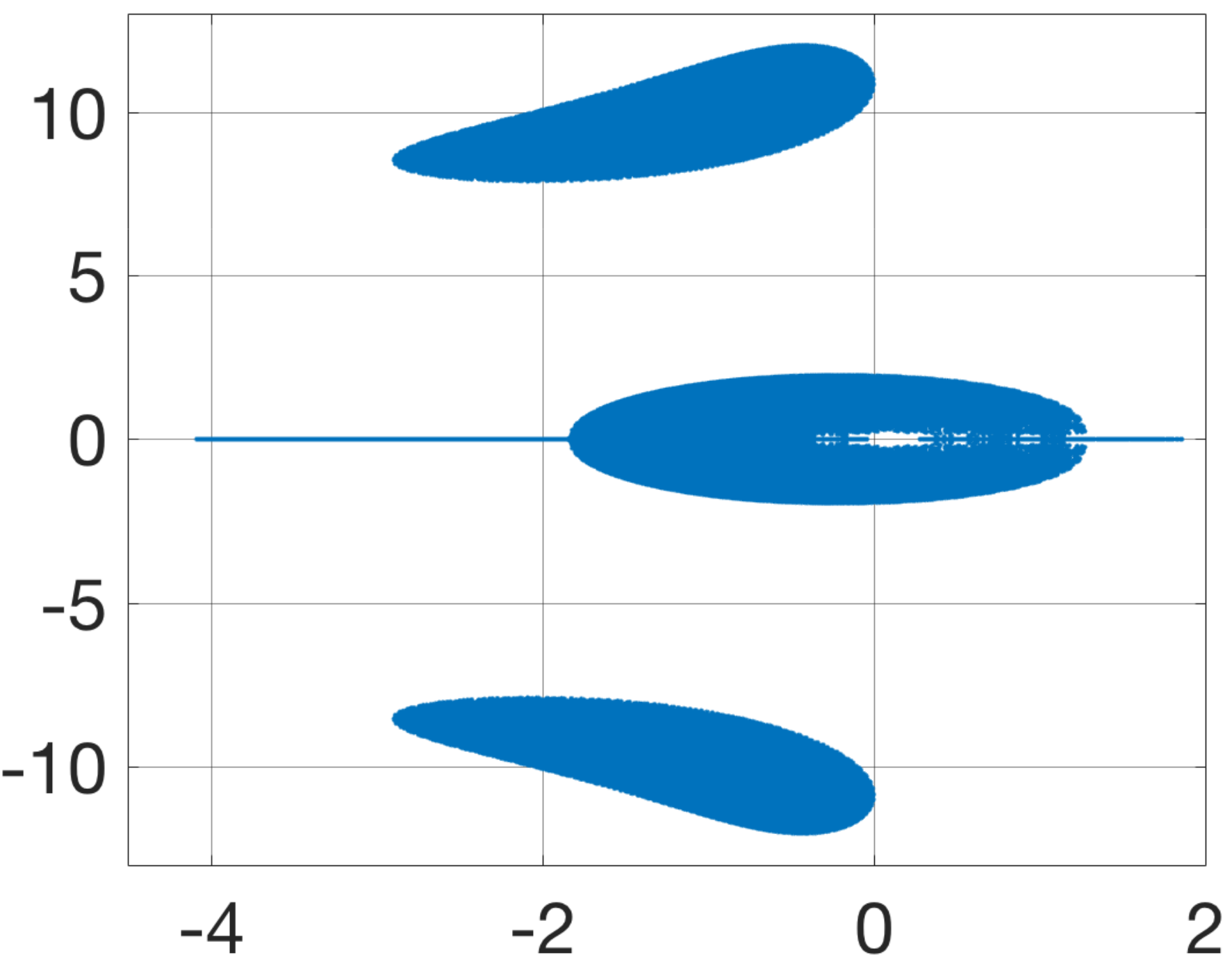} \label{fig:ss2}}
  \subfigure[$\eta = 0.5653$]{
  \includegraphics[width=.45\columnwidth]{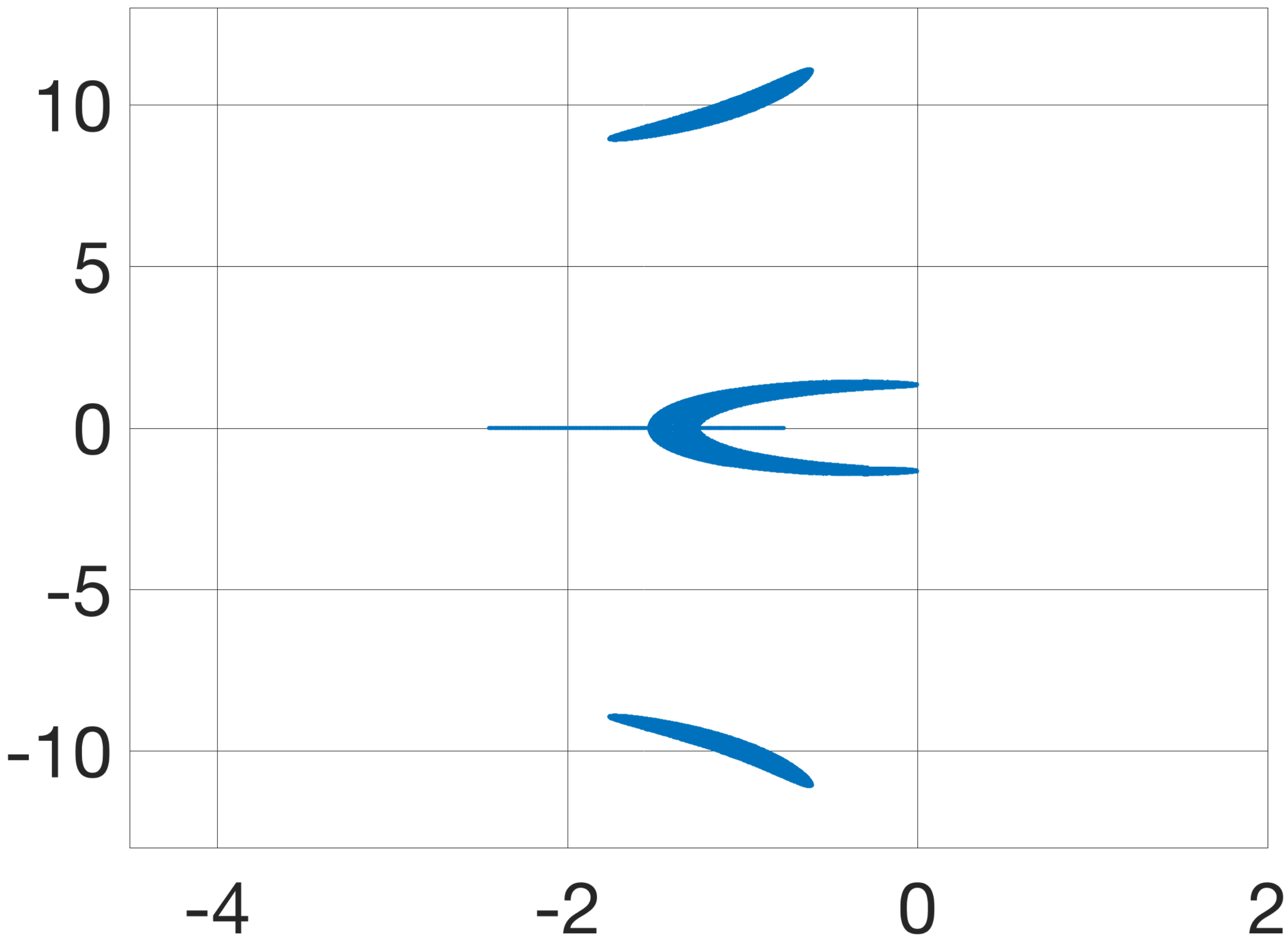} \label{fig:ss3}}
  \subfigure[$\eta = 4.9622$]{
  \includegraphics[width=.45\columnwidth]{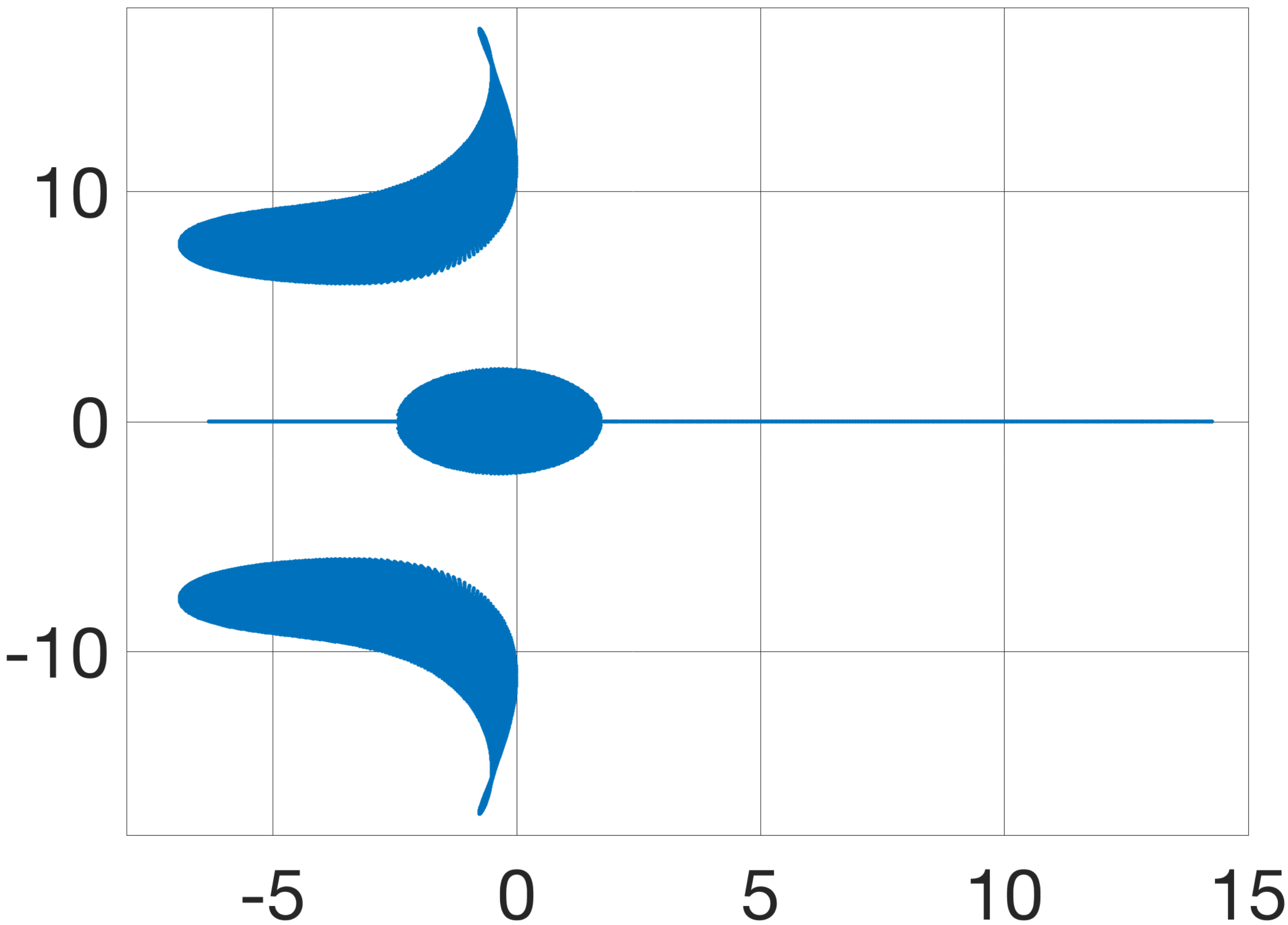} \label{fig:ss4}}
  \caption{The spectral value sets corresponding to the local minima in Table \ref{tab:local_min}. Figures (a)-(b) correspond to Case 1, and (c)-(d) correspond to Case 2.}
  \label{fig:spec_sets}
\end{figure}
  
Figure \ref{fig:iter} presents a sample run of Algorithm \ref{algo:grad_New_des} for the case 2 of the above example. It is initialized with $\omega_0 = 2.5$ and $g_0 = [1.0582, 0.4363, 1.4115, -0.0146]^{\trans}$ and takes $24$ iterations to converge to the global minimum using the Newton descent steps. Figure \ref{fig:iter} shows the penalized cost (scaled), spectral abscissa of the perturbed matrix $A(\Delta)$, and $\omega$ at each iteration. Observe that, at the start of the algorithm, $\alpha(A(\Delta))>0$ indicating that $A(\Delta)$ has two eigenvalues in the right-half complex plane (the other two are at $\pm j\omega$). As the iterations progress, these unstable eigenvalues move towards the left-half plane and, at the global minimum, all eigenvalues are in the closed left-half plane (c.f. discussion after \eqref{eq:spar_const1}). Further, the optimization cost decreases monotonically during the iterations.

\begin{figure}[h!]
\centering
\includegraphics[width=\columnwidth]{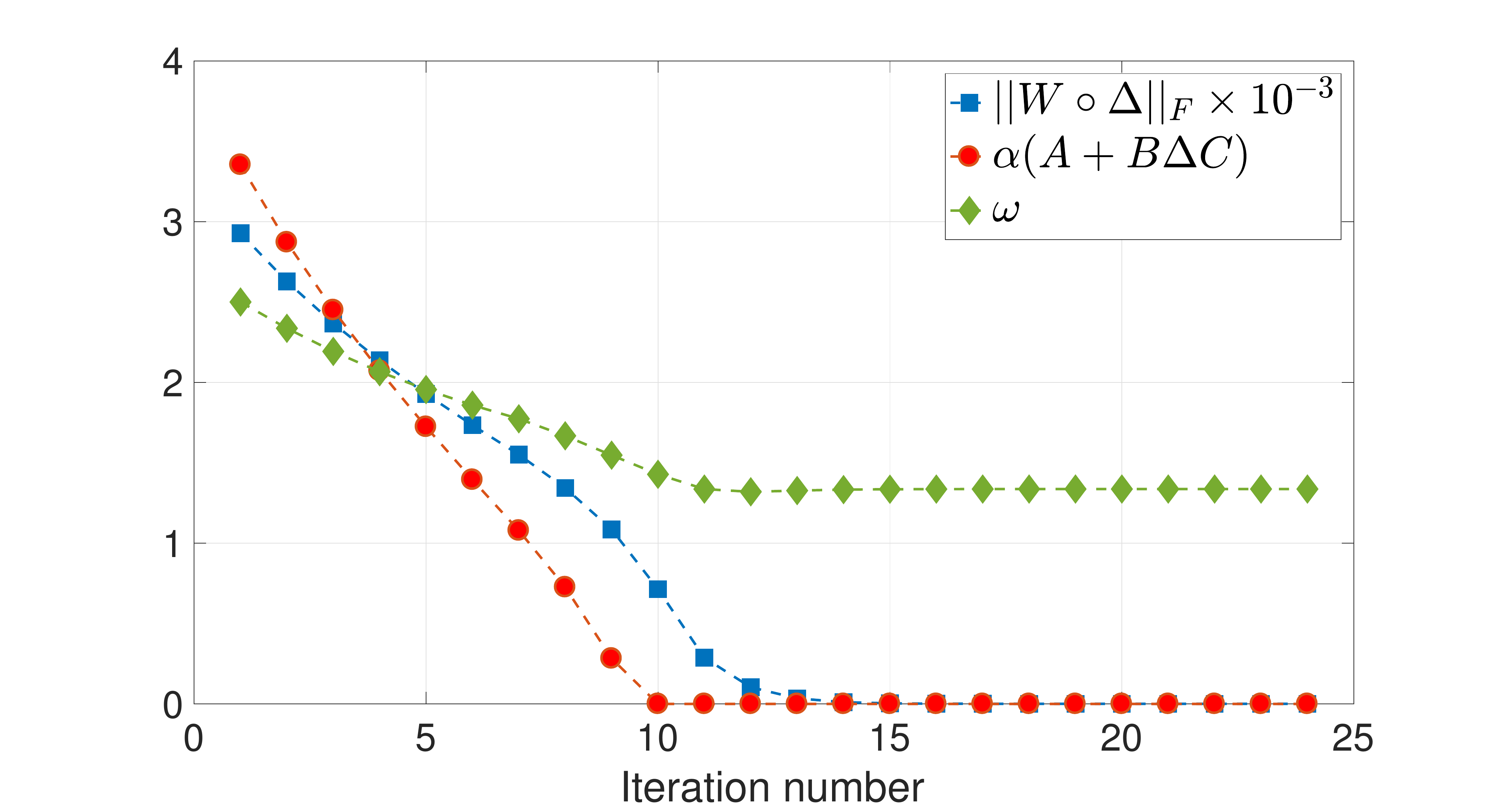}
\caption{A sample iteration run of Algorithm \ref{algo:grad_New_des}.}
\label{fig:iter}
\end{figure}

Next, we present a comparison of the global minima obtained by Algorithm \ref{algo:grad_New_des} for different penalty weights. Table \ref{tab:approx-spar} shows the global minima for three values of weight ${\bf w}$. Further, Figure \ref{fig:approx-spar} shows the sparsity error $E \triangleq ||\Delta - S \circ \Delta||_F$ and the norm of the optimal perturbation as a function of the weight $w$. Observe that, as the weight ${\bf w}$ increases, the sparsity error decreases and the optimal perturbations become more sparse. Furthermore, the norm of the optimal perturbations increases with ${\bf w}$, since a larger weight implies a tighter constraint on the perturbation entries.  

\begin{table}[h!]
\caption{Approximately-sparse Solutions}
\begin{center}
\setlength{\tabcolsep}{10pt}
 
  \begin{tabular}{| c | c |c|c| }
    \hline
${\bf w}$  & $\hat{\Delta}$ & $||\hat{\Delta}||_F$ & $\hat{\omega}$ \\ \hline 
\rule{0pt}{15pt}$5$ &
$\begingroup % keep the change local
\setlength\arraycolsep{3pt}
\begin{bmatrix}
 -0.0414 &  -0.0036\\
    \phantom{-}0.0095  &  \phantom{-} 0.5593
\end{bmatrix} \endgroup$ &  $ 0.5609$ & $1.3385$ \\ [1em]\hline

\rule{0pt}{15pt}$10$ &
$\begingroup % keep the change local
\setlength\arraycolsep{3pt}
\begin{bmatrix}
-0.0417 &  -0.0009\\
    \phantom{-} 0.0024  &  \phantom{-} 0.5627
\end{bmatrix}\endgroup$ &  $0.5642$ &  $1.3370$ \\[1em] \hline

\rule{0pt}{15pt}$20$ &
$\begingroup % keep the change local
\setlength\arraycolsep{3pt}
\begin{bmatrix}
-0.0418 &  -0.0002\\
  \phantom{-}   0.0006  &  \phantom{-} 0.5635
\end{bmatrix}\endgroup$ &  $0.5651$ &  $1.3367$ \\[1em] \hline

\end{tabular}
\end{center}
\label{tab:approx-spar}
\end{table}
%width=0.85\columnwidth, height = 1.5in

%\begin{table}[h!]
%\caption{Approximately-sparse Solutions}
%\begin{center}
%\setlength{\tabcolsep}{2pt}
% 
%  \begin{tabular}{| c | c | }
%    \hline
%
%\rule{0pt}{15pt}$w = 5$ &
%$\hat{\Delta}= 
%\begingroup % keep the change local
%\setlength\arraycolsep{3pt}
%\begin{bmatrix}
% -0.0414 &  -0.0036\\
%    \phantom{-}0.0095  &  \phantom{-} 0.5593
%\end{bmatrix} \endgroup$\!,  $||\hat{\Delta}||_F = 0.5609$,  $\hat{\omega} =1.3385$ \\ [1em]\hline
%
%\rule{0pt}{15pt}$w = 10$ &
%$\hat{\Delta}= 
%\begingroup % keep the change local
%\setlength\arraycolsep{3pt}
%\begin{bmatrix}
%-0.0417 &  -0.0009\\
%    \phantom{-} 0.0024  &  \phantom{-} 0.5627
%\end{bmatrix}\endgroup$\!,  $||\hat{\Delta}||_F = 0.5642$,  $\hat{\omega} =1.3370$ \\[1em] \hline
%
%\rule{0pt}{15pt}$w = 20$ &
%$\hat{\Delta}= 
%\begingroup % keep the change local
%\setlength\arraycolsep{3pt}
%\begin{bmatrix}
%-0.0418 &  -0.0002\\
%  \phantom{-}   0.0006  &  \phantom{-} 0.5635
%\end{bmatrix}\endgroup$\!,  $||\hat{\Delta}||_F =0.5651$,  $\hat{\omega} =1.3367$ \\[1em] \hline
%
%\end{tabular}
%\end{center}
%\label{tab:approx-spar}
%\end{table}
%%width=0.85\columnwidth, height = 1.5in

\begin{figure}[h!]
  \centering
  \subfigure[]{
 \hspace{1pt} \includegraphics[width=.75\columnwidth]{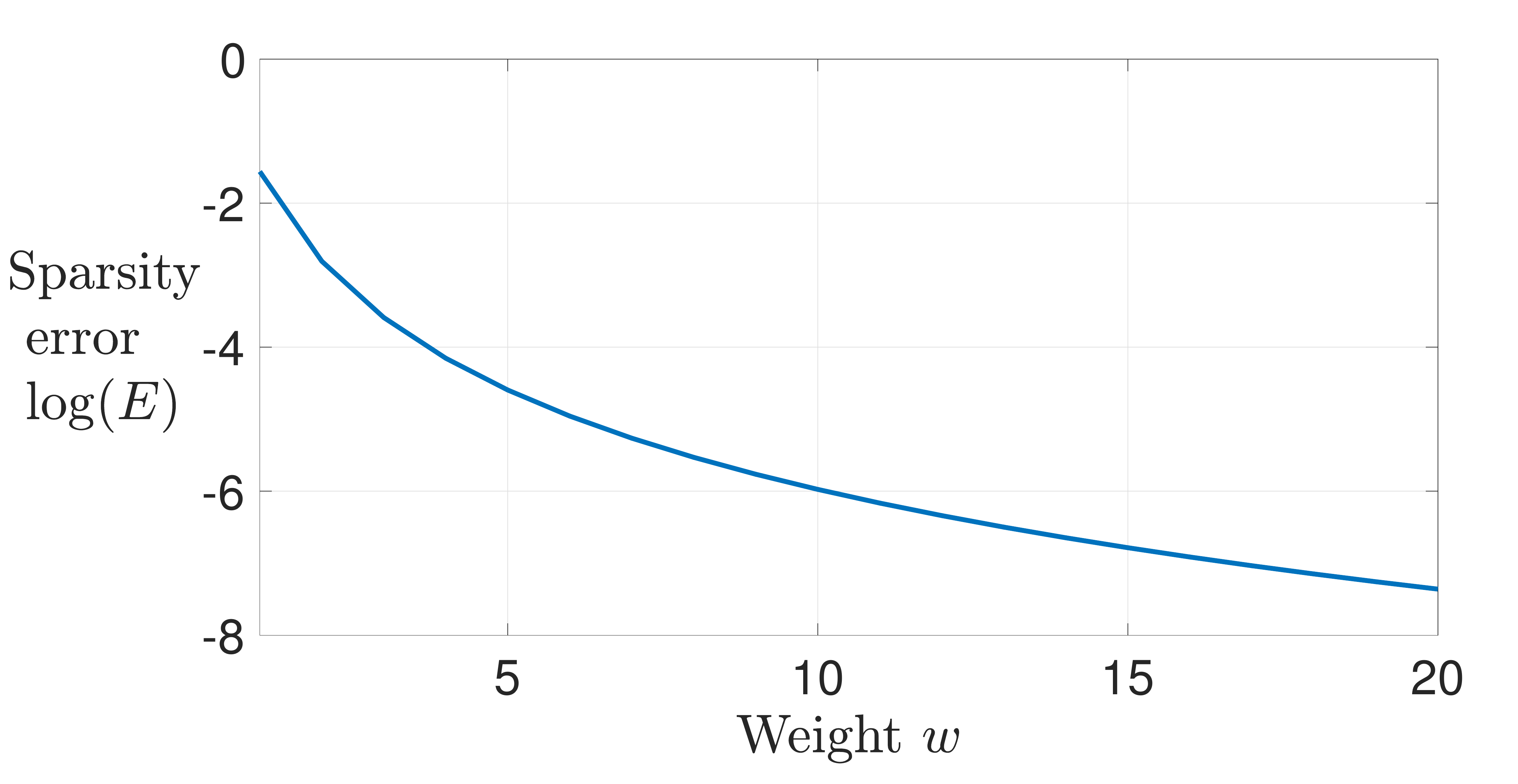} \label{fig:approx-spar_a}}
  \subfigure[]{
 \includegraphics[width=.75\columnwidth]{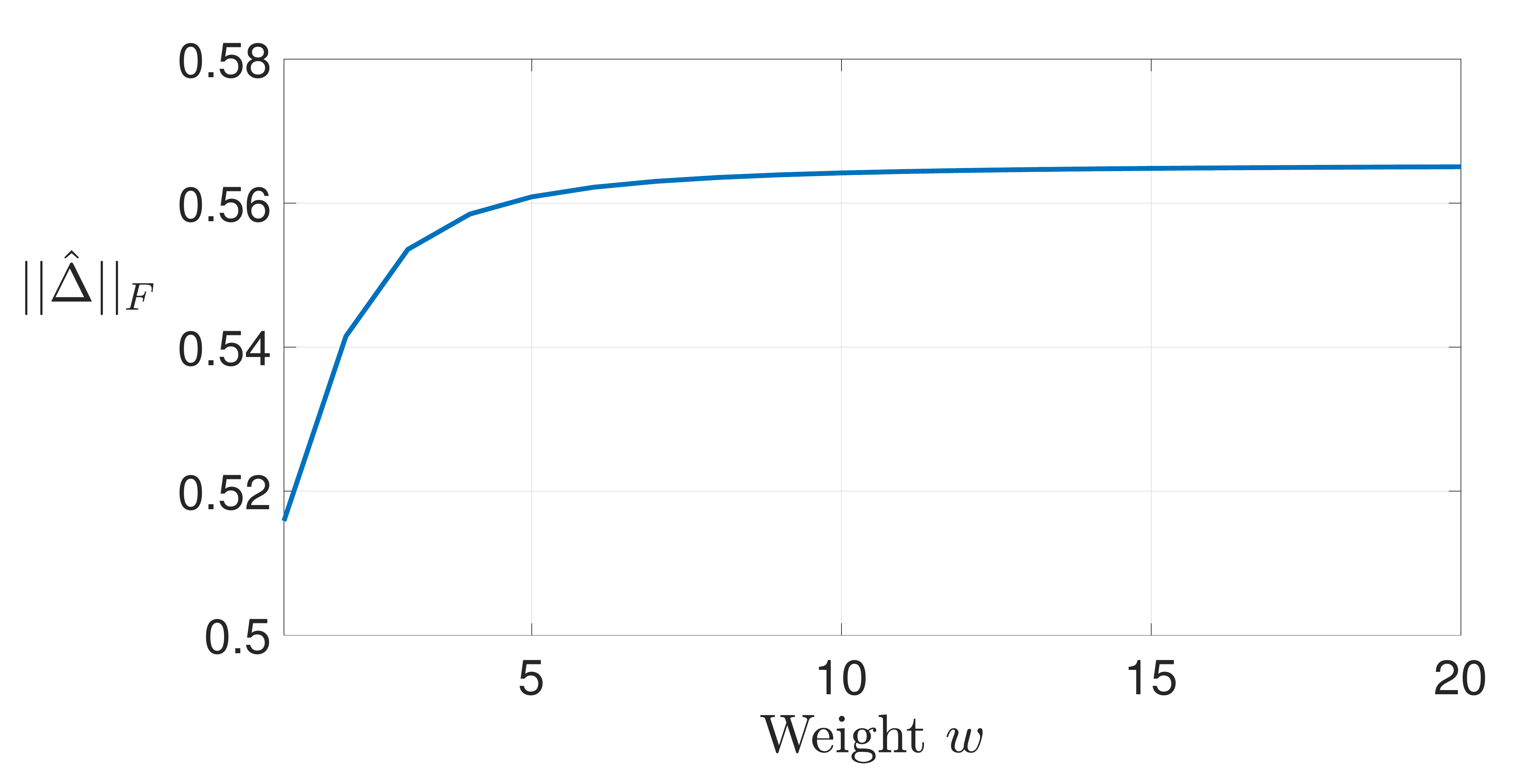} \label{fig:approx-spar_b}}
  \caption{Variation of (a) sparsity error $\log(E)$, and (b) norm of the optimal perturbation $\hat{\Delta}$, as a function of weight $w$.}
  \label{fig:approx-spar}
\end{figure}

\begin{figure}[h!]
  \centering
  \subfigure[Line network]{
 \hspace{1pt} \includegraphics[width=.75\columnwidth]{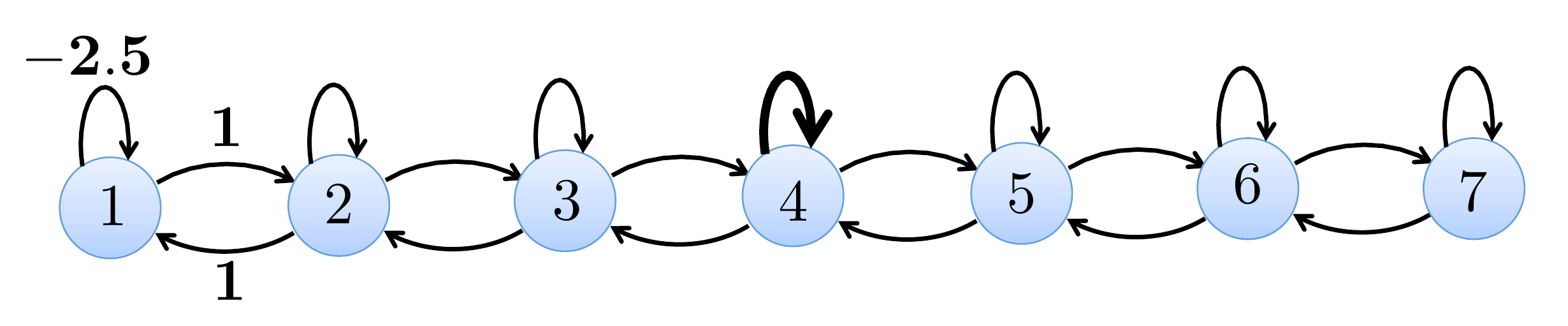} \label{fig:line_netw}}
  \subfigure[Circular network]{
 \includegraphics[width=.75\columnwidth]{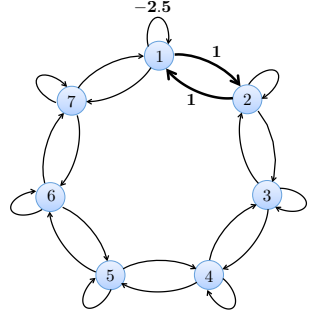} \label{fig:circ_netw}}
  \caption{Two symmetric networks. The bold edge(s) is/are most critical and result in the smallest sparse SR.}
  \label{fig:networks}
\end{figure}

Finally, we illustrate that our sparse SR framework provides structural insights into the stability of dynamical networks. We consider symmetric line and circular networks as shown in Figure \ref{fig:networks}, where the nodes represent the scalar states and the edges represent the non-zero couplings. All self loops have weight $-2.5$ and all inter-node edges have weight $1$. The state matrix $A$ can be easily constructed using these weights and it is stable.

We are interested in identifying the edge(s) that are most critical for the stability of the network. This can be characterized by assigning a sparsity pattern corresponding to a subset of edges that are perturbed, and computing the sparse SR using the developed framework. Then, the most critical edge set is the one which results in the least SR. For the line network, we allow only a singe edge to be perturbed. It implies that only one entry of $\Delta$ is allowed to be perturbed. For the circular network, we allow for two inter-node edges to be perturbed (self loop edges are fixed). This implies that only two non-diagonal entries of $\Delta$ are allowed to be perturbed. We set $B=C=I$ for both the networks.

For the line network, we observe that the most critical edge is the self loop of the node in the center of the line (node $4$).\footnote{If there is an even number of nodes, then there are two most critical edges corresponding to the self loops of the two center nodes.} The SR corresponding to this edge is $1.5118$. For the circular network, the two most critical inter-node edges are the edges between any two neighboring nodes. The optimal perturbations for the two edges are $0.9724$ and $0.9814$ (in any order) and the corresponding SR is $1.3816$. Due to the circular symmetry, there exist $7$ pairs of critical edges in the network. These examples highlight that our sparse SR framework is useful in studying the robust stability of sparse networks, which was not possible using the previous non-sparse SR theory. 
%As a future work, we plan to extend this structural analysis to more general sparse networks. 

%Given a bound on the number of edges that are allowed to be changed, we are interested in finding the edges of a network whose perturbations results in the smallest SR. This exercise provides us with the most critical edges  that can result in instability of the network. For the line network, we allow only a single edge to be perturbed, while for the circular network, we allow two non-self loops to be perturbed (self loops are fixed). We encode these allowable sparsity pattern into the matrix $S$ can set $B=C=I_7$. For the line network, we observe that the most critical edge is the self coupling edge of node in the middle of the line (node $4$). \footnote{If there are even nodes in the network, then there are two possible critical edges.} The SR corresponding to this edge is $0.56$. For the circular network case, the two most critical inter-node edges are separated by ..... . The SR corresponding to these edges is $0.56$. 

\section{Conclusion} \label{sec:conclusion} 

In this paper we study the real, sparse, $F$-norm stability radius
of a linear time-invariant system, which measures its ability to
maintain stability in the presence of structured additive perturbations. We
formulate the stability radius problem as an equality-constrained
minimization problem, and characterize its optimality
conditions. These conditions reveal important geometric properties of
the stability radius and the associated perturbation, and allow us to
design a penalty based Newton descent algorithm that provably converges to locally
optimal values of the stability radius and the associated
perturbation. Using the Frobenius norm to measure the size of
perturbations is not only convenient for the analysis, but it also
provides selective information regarding which system entries have a
greater effect on system stability. Further, imposing an arbitrary
sparsity pattern to the perturbation becomes crucial when
studying the stability radius of network systems and, more generally,
systems where only a subset of the entries can be perturbed. Numerical
examples are shown to highlight the utility of our framework for characterizing
structural fragility of networks. 

% We studied the real, structured, F-norm stability radius problem in
% this paper. We formulated the SR problem as an equality-constrained
% optimization problem and computed its optimality conditions. The use
% of Frobenius norm allowed us to compute analytic expressions of the
% gradients and develop a gradient descent algorithm with guaranteed
% convergence to a local minima. We also used our optimization based
% approach to obtain approximately-sparse solutions of the SR problem
% with arbitrary sparsity constraints, a previously unsolved problem, by
% suitably weighing the entries of the perturbations. We plan to extend
% the analysis to obtain exactly sparse solutions to the SR problem,
% which are more suitable to large scale networked and sparse systems.

\renewcommand{\baselinestretch}{0.982}

%\bibliographystyle{unsrt}
%\bibliography{./bib/alias,./bib/Main,./bib/New,./bib/FP}

\end{document}